\documentclass[11pt]{article}

\usepackage[T1]{fontenc}
\usepackage[utf8]{inputenc}
\usepackage{lmodern}
\usepackage{microtype}
\usepackage[margin=0.92in]{geometry}
\usepackage{amsmath,amssymb,amsthm,mathtools,bm}
\usepackage{booktabs,tabularx,longtable,array,multirow}
\usepackage{enumitem}
\usepackage{xcolor}
\usepackage{graphicx}
\usepackage{caption}
\usepackage{subcaption}
\usepackage{tikz}
\usetikzlibrary{arrows.meta,positioning,calc,fit,backgrounds,shapes.geometric,matrix}
\usepackage{float}
\usepackage{lscape}
\usepackage[numbers,sort&compress]{natbib}
\usepackage{hyperref}
\usepackage[nameinlink,capitalize,noabbrev]{cleveref}
\usepackage{fancyhdr}
\usepackage{etoolbox}
\usepackage{aliascnt}

\definecolor{stqpblue}{RGB}{27,92,141}
\definecolor{stqpgreen}{RGB}{43,122,91}
\definecolor{stqporange}{RGB}{202,112,39}
\definecolor{stqpgray}{RGB}{245,247,249}
\definecolor{stqpdark}{RGB}{43,48,54}
\definecolor{stqpred}{RGB}{157,52,52}

\hypersetup{
  colorlinks=true,
  linkcolor=stqpblue,
  citecolor=stqpgreen,
  urlcolor=stqporange,
  pdftitle={Standard Quadratic Formulations of Many NP Problems},
  pdfauthor={Mohammad-Ali Miri, Babak Emami, PoJen Wang}
}

\floatstyle{ruled}
\newfloat{algorithm}{htbp}{loa}
\floatname{algorithm}{Algorithm}
\crefname{algorithm}{Algorithm}{Algorithms}
\Crefname{algorithm}{Algorithm}{Algorithms}

\setlist[itemize]{topsep=3pt,itemsep=2pt,parsep=0pt}
\setlist[enumerate]{topsep=3pt,itemsep=2pt,parsep=0pt}
\renewcommand{\arraystretch}{1.14}
\newcolumntype{L}[1]{>{\raggedright\arraybackslash}p{#1}}

\newtheorem{theorem}{Theorem}[section]
\newcommand{\sharedthm}[2]{%
  \newaliascnt{#1}{theorem}\newtheorem{#1}[#1]{#2}\aliascntresetthe{#1}%
  \crefname{#1}{#2}{#2s}\Crefname{#1}{#2}{#2s}}
\sharedthm{proposition}{Proposition}
\sharedthm{lemma}{Lemma}
\sharedthm{corollary}{Corollary}
\theoremstyle{definition}
\sharedthm{definition}{Definition}
\sharedthm{example}{Example}
\sharedthm{construction}{Construction}
\theoremstyle{remark}
\sharedthm{remark}{Remark}
\sharedthm{caution}{Caution}

\newcommand{\R}{\mathbb{R}}

\newcommand{\one}{\bm{1}}

\newcommand{\StQP}{\mathrm{StQP}}

\newcommand{\norm}[1]{\left\lVert #1\right\rVert}
\newcommand{\abs}[1]{\left\lvert #1\right\rvert}
\newcommand{\set}[1]{\left\{#1\right\}}

\newcommand{\Qhat}{\widehat Q}

\newenvironment{thesisbox}{%
  \begin{center}\begin{minipage}{0.93\linewidth}
  \color{stqpdark}\hrule\vspace{0.65em}\noindent
}{%
  \par\vspace{0.65em}\hrule\end{minipage}\end{center}
}

\title{\bfseries Standard Quadratic Formulations of Many {NP} Problems:\\[3pt]
A Simplex-Based Compilation Framework for Combinatorial Optimization}
\author{Mohammad-Ali Miri$^{*}$, Babak Emami, PoJen Wang\\
\small Quantum Computing Inc (QCi), 5 Marine View Plaza, Hoboken, NJ 07030, USA\\
\small \texttt{amiri@quantumcomputinginc.com}}
\date{September 29, 2026}

\begin{document}
\maketitle

\begin{abstract}
The standard quadratic program (StQP) minimizes a quadratic form over nonnegative variables that sum to one. We compose classical graph reductions with regularized Motzkin--Straus clique formulations to express discrete optimization problems in this continuous domain. The graph matrix has diagonal entries $\tau$, zeros on edges, and ones on nonedges. For $0<\tau<1$, its minimum is $\tau/\omega(G)$, where $\omega(G)$ is the clique number. Its strict local minimizers are precisely the uniform distributions on maximal cliques, and its global minimizers encode maximum cliques. At $\tau=1/2$, integer scaling gives coefficients in $\{0,1,2\}$ and minimum $1/\omega(G)$, yielding an NP-complete StQP threshold problem with a restricted coefficient alphabet. We give explicit formulations for satisfiability, coloring, Hamiltonian cycles, independent set, vertex cover, set packing, three-dimensional matching, and graph isomorphism. A regularized weighted clique formulation combined with local-state compatibility graphs gives an exact compiler for finite-domain factor models specified by complete local tables, including QUBO, with at most four simplex coordinates per binary pair factor. The catalog covers Karp's 21 problems: twelve use direct graph formulations, and nine use factor-state formulations, including six obtained through binary-linear feasibility. For each route we record dimensions, coefficient structure, and recovery rules. We analyze interaction count, coefficient range, objective separation, perturbation tolerance, support recovery, and decoding overhead. The separation bounds quantify the effects of clique size, factor weights, and offsets. In the complete factor-state construction, every assignment, including each suboptimal assignment, is a strict local minimum.
\end{abstract}

\noindent\textbf{Keywords:} standard quadratic programming; simplex; NP-completeness; maximum clique; constraint satisfaction; QUBO; Ising model; combinatorial optimization; analog optimization; simplex quadratic annealer.

\clearpage
\begingroup
\small
\tableofcontents
\endgroup
\clearpage

\section{Introduction}
\label{sec:intro}

The Ising Hamiltonian and its equivalent quadratic unconstrained binary optimization (QUBO) form serve as common compilation targets for optimization algorithms and hardware. A large body of work maps logical, graph, packing, covering, routing, assignment, and scheduling problems to pairwise binary interactions, and those mappings have helped define the interface of quantum annealers, coherent Ising machines, simulated bifurcation systems, and related platforms \citep{Barahona1982,KochenbergerEtAl2014,Lucas2014,GloverKochenbergerDu2019,MarandiEtAl2014,McMahonEtAl2016,goto-simulated-bifurcation-2019,mohseni-mcmahon-byrnes-ising-machines-2022,tasseff-emerging-potential-qa-2024}. Lucas's catalog of Ising formulations for many NP problems, including Karp's 21 classical NP-complete problems, is an especially influential expression of this viewpoint \citep{Karp1972,Lucas2014}.

That interface is binary by construction. For a device whose native state is a normalized, nonnegative vector, such as an intensity or population profile, a simplex formulation may provide a more direct interface. Converting to a binary representation can introduce discretization, one-hot constraints, or consistency penalties, depending on the encoding; these costs are not inherent to every Ising model. This paper develops a formulation catalog for the simplex and compares its costs with those of binary alternatives.

The continuous model considered here is
\begin{equation}
  \nu(Q):=\min_{x\in\Delta_N}x^{\mathsf T}Qx,
  \qquad
  \Delta_N:=\set{x\in\R^N:x\ge 0,\ \one^{\mathsf T}x=1},
  \label{eq:stqp}
\end{equation}
where $Q\in\mathbb S^N$. This is the standard quadratic program (StQP) \citep{Bomze1998}. The state is continuous but compact, nonnegative, and normalized. A coordinate may be interpreted as mass, probability, population, intensity, or allocation, and $Q_{uv}$ measures a pairwise interaction between coordinates $u$ and $v$.

At first sight, replacing binary variables by simplex variables seems poorly suited to exact combinatorial encoding. Normalization is what makes it work: it turns positive cross terms into incompatibility penalties. If two coordinates represent mutually inconsistent local choices, assigning mass to both incurs a quadratic cost. If compatible choices interact with coefficient zero, global minimization seeks a mutually compatible support. A diagonal term controls the concentration of mass on that support. For a carefully chosen three-level matrix, the optimal support is a maximum clique and the optimal vector is its barycenter.

In the formulations studied here, discrete choices are recovered from the support of a normalized continuous state. The matrix specifies incompatible pairs, while diagonal coefficients determine the optimum on each compatible support. We compare the dimension, connectivity, precision, and decoding costs of this representation with those of problem-specific binary formulations.

\subsection{Scope and contributions}

The regularized unweighted and weighted clique formulations provide the mathematical foundations \citep{Bomze1997,HungerfordRinaldi2019,BomzePelilloStix2000}. We state and prove them in a common minimization normalization so that the matrix entries, objective values, and decoding rules can be used throughout the catalog. The paper develops four aspects of this formulation framework.

\begin{enumerate}[label=\textbf{C\arabic*.},leftmargin=2.3em]
\item \textbf{Explicit formulations and a complete catalog.} We compose classical compatibility-graph constructions with regularized clique StQP formulations. For Karp's 21 problems, twelve routes use direct unweighted or weighted graphs and nine use factor-state formulations, six of them through binary-linear feasibility. Each entry specifies its route, dimension or encoding-bit count, decision or recovery rule, and principal representation cost.

\item \textbf{A common factor-state formulation and its landscape.} We combine positive local-state weights, consistency graphs, and a weighted clique matrix with binary off-diagonal entries. For complete rational factor tables the construction is polynomial in their explicit encoding length; a binary model with $n_u$ unary factors and $m$ pair factors uses at most $2n_u+4m$ simplex coordinates. We give an exact recovery map and show that every complete assignment becomes a strict local minimum (\Cref{cor:allassignments}).

\item \textbf{Objective separation and decoding guarantees.} We record the NP-complete threshold problem associated with the coefficient alphabet $\{0,1,2\}$ and derive graph decision-value separations, perturbation tolerances, and factor-model separation bounds. We analyze how offsets and attainable energy spacings enter these bounds, and connect objective accuracy to a specified decoding procedure.

\item \textbf{Representation-cost comparison.} We compare simplex dimension, interaction count, coefficient range, objective resolution, support recovery, and decoding overhead with strong problem-specific binary formulations. These quantities define the costs that a computational evaluation must measure.
\end{enumerate}

\subsection{Relation to prior work}

The reduction of NP decision problems to StQP follows by composing the classical reduction to CLIQUE \citep{Cook1971,Karp1972} with the Motzkin--Straus formulation and its regularizations \citep{MotzkinStraus1965,Bomze1997,GibbonsEtAl1997,HungerfordRinaldi2019}. StQP itself is a well-established class in global optimization \citep{Bomze1998,MurtyKabadi1987,BomzeDeKlerk2002,BomzeSchachingerUllrich2018}. The direct SAT, coloring, Hamiltonian, and packing routes use classical reductions and compatibility constructions; composing them with the clique StQP supplies explicit matrices and thresholds rather than new representability results.

Compatibility and microstructure graphs are classical in constraint processing \citep{Jegou1993,Dechter2003,RossiVanBeekWalsh2006}. Association graphs underlie matching and subgraph isomorphism \citep{barrow-burstall-1976,levi-mcs-1973}, and dual microstructures yield the allowed-tuple construction for arbitrary-arity constraints \citep{elmouelhi-microstructures-arity-2013}. Reductions of finite-domain maximum a posteriori estimation to maximum-weight independent set also predate this work \citep{SanghaviShahWillsky2007}. The factor-state compiler combines this type of consistency construction with a weighted clique StQP. Weighted and regularized clique formulations have substantial precedents \citep{BomzePelilloStix2000,Busygin2006,pardalos-prokopyev-busygin-2006,hager-hungerford-graphqp-2015,makovenko-butenko-pardalos-regularized-2025,scozzari-tardella-clique-algorithm-2008,tardella-ftlp-2004}; the normalization comparison after \Cref{thm:weightedlandscape} specifies the relation to \citet{BomzePelilloStix2000}.

Two formulation catalogs provide useful points of comparison. \citet{Lucas2014} uses binary Ising interactions, whereas \citet{cai-oscillator-np-2025} develops first-order and multiphase Potts formulations and a three-state oscillator optimizer. The present catalog targets a quadratic objective over one globally normalized, nonnegative state. Its emphasis is the explicit matrix and decoding description, together with dimension, coefficient, separation, and landscape costs; it does not introduce an oscillator architecture or establish solver performance.

Relaxation labeling provides another relevant comparison. Its per-object label distributions have separate normalization constraints, with quadratic compatibility objectives and well-studied dynamics \citep{Pelillo1997Relaxation}. Such products of simplexes differ from the single simplex used here. \Cref{subsec:alternative-lifts} compares these feasible sets and explains why their local-minimum structures need not agree.

\subsection{Computational scope}

An exact formulation does not imply an efficient solution method. Unless $\mathrm P=\mathrm{NP}$, no polynomial-time algorithm can solve every encoded instance. Nor does a compact coefficient alphabet guarantee a large physical energy gap, robust analog behavior, or quantum advantage. The local-minimum correspondences characterize the optimization landscape; they do not determine the convergence rate or success probability of a solver.

We use \emph{StQP} for the mathematical optimization model. When discussing a hypothetical hardware platform, we use \emph{simplex quadratic optimizer/annealer} and abbreviate it as \emph{StQA}. Nothing in the formulations assumes that annealing is the only or best physical solution mechanism.

\subsection{Organization}

\Cref{sec:prelim} defines exactness and representation metrics. \Cref{sec:graphcompiler} develops the graph compiler and its complexity and precision consequences. \Cref{sec:csp} gives direct compatibility formulations. \Cref{sec:weightedcompiler} develops the weighted factor-state and QUBO compilers. \Cref{sec:packing} treats graph selection, packing, and exact cover. \Cref{sec:bilp} gives a generic binary-linear route. \Cref{sec:karp} catalogs Karp's 21 problems. \Cref{sec:cost} compares representation costs with QUBO/Ising. \Cref{sec:hardware} states a hardware-neutral StQA specification and evaluation metrics. \Cref{sec:limitations} collects the mathematical, design, and empirical limitations, and \Cref{sec:conclusion} concludes. The appendices give compilation procedures (\Cref{app:pseudocode}), decoding and purification (\Cref{app:decoding}), explicit formulations for selected catalog entries (\Cref{app:formulations}), the copositive connection (\Cref{app:copositive}), and notation (\Cref{app:notation}).

\section{Preliminaries and formulation standards}
\label{sec:prelim}

\subsection{The standard quadratic program}

For $Q\in\mathbb S^N$, define
\begin{equation}
  \StQP(Q):\qquad \nu(Q)=\min_{x\in\Delta_N}x^{\mathsf T}Qx.
  \label{eq:stqpdef}
\end{equation}
The simplex $\Delta_N$ has $N$ vertices and affine dimension $N-1$. We use the convention that the superscript or subscript records the number of coordinates, not the intrinsic dimension. We write $\mathbb S^N$ for the space of real symmetric
$N\times N$ matrices, $\mathbf 1$ for the all-ones vector, $I_N$ for the identity matrix, and $e_i$ for the $i$th coordinate vector. We use $[n]=\{1,\ldots,n\}$. For $x\in\Delta_N$, its support is $\operatorname{supp}(x):=\{i:x_i>0\}$. A nonempty graph has at least one vertex; it need not have an edge. Empty instances, empty domains, and trivial or out-of-range decision thresholds are handled before compilation so that the simplex and the displayed reciprocal thresholds are well defined.

The model is invariant under several elementary transformations. Replacing $Q$ by $(Q+Q^{\mathsf T})/2$ does not alter the objective. Adding $\beta\one\one^{\mathsf T}$ shifts every feasible objective by $\beta$ and leaves the minimizers unchanged. Multiplication by a positive scalar preserves minimizers and rescales all objective gaps. These operations are important when selecting a hardware-friendly coefficient alphabet.

\begin{definition}[Encoding]
Let $\mathcal I$ be an instance of a discrete decision or optimization problem. An \emph{StQP encoding} consists of a polynomial-time map $\mathcal I\mapsto(Q_{\mathcal I},\theta_{\mathcal I})$, together with a polynomial-time decoder, such that the desired answer can be recovered from $\nu(Q_{\mathcal I})$, from a global minimizer, or from both. Here $\theta_{\mathcal I}$ denotes any threshold or affine recovery data.
\end{definition}

\begin{definition}[Levels of exactness]
\label{def:exactness}
We distinguish the following properties.
\begin{enumerate}[label=(\roman*)]
\item \emph{decision exactness}: a yes/no answer is equivalent to an objective threshold;
\item \emph{value exactness}: the original optimum is recovered by an explicit function of $\nu(Q)$;
\item \emph{solution exactness}: a global minimizer can be decoded to an original optimum;
\item \emph{landscape exactness}: relevant local optima or barriers correspond under the mapping;
\item \emph{approximation preservation}: a stated objective error in StQP implies a controlled error in the original problem.
\end{enumerate}
\end{definition}

Most reductions in complexity theory require only decision exactness. Hardware evaluation needs more. A formulation may be perfectly exact at the global optimum while creating exponentially many local minima, compressing decision values into a narrow interval, or requiring delicate support thresholding.

\subsection{Representation metrics}
\label{subsec:metrics}

For an encoded instance, we report:
\begin{align*}
  N_\Delta &:= \text{number of simplex coordinates},\\
  M_Q &:= \#\set{(u,v):u\le v,\ Q_{uv}\ne 0},\\
  d_Q &:= \max_u \#\set{v\ne u:Q_{uv}\ne 0},\\
  \mathcal A_Q &:= \set{Q_{uv}:1\le u\le v\le N_\Delta},\\
  \kappa_{\mathrm{coef}} &:= \frac{\max\set{|q|:q\in\mathcal A_Q}}{\min\set{|q|:q\in\mathcal A_Q\setminus\{0\}}}.
\end{align*}
Additional quantities are the target support size, objective separation between adjacent decision classes, entrywise perturbation tolerance, decoder complexity, and number or structure of local minima.

The raw nonzero count should be interpreted carefully. Since $\one^{\mathsf T}x=1$, adding a constant all-to-all matrix $\beta\one\one^{\mathsf T}$ is physically and mathematically redundant up to an objective shift. A platform may implement such a baseline globally rather than through $O(N^2)$ individual couplers. We therefore distinguish the literal matrix $Q$ from its \emph{relative interaction matrix}, obtained after removing a convenient multiple of $\one\one^{\mathsf T}$.

\subsection{Decision thresholds and certificates}

For graph-coded instances in \Cref{sec:graphcompiler}, an optimal support itself gives a polynomial certificate. For general rational StQP, complexity membership requires more care because minimizers are real. Quadratic programming with rational data admits polynomially verifiable certificates under standard formulations \citep{Vavasis1990,Vavasis1991}, but our NP-completeness statement will avoid unnecessary subtleties by restricting to an explicit graph matrix class that is exactly equivalent to CLIQUE.

\subsection{Compilation pipeline}

\Cref{fig:pipeline} summarizes the routes developed below. Direct routes preserve combinatorial structure and usually give the cleanest coefficient matrices. Generic routes trade compactness for universality.

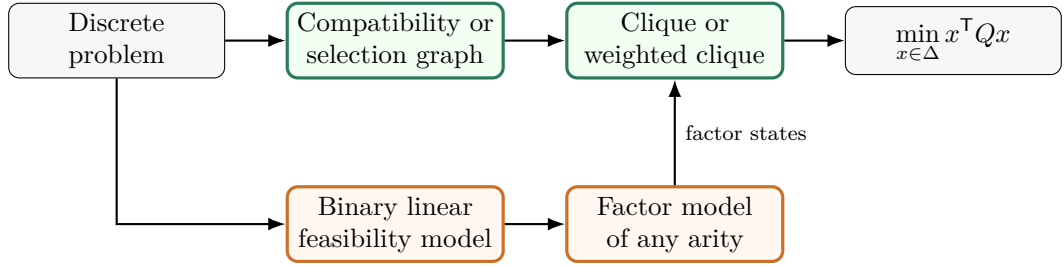
\begin{figure}[t]
\centering
\begin{tikzpicture}[
  node distance=9mm and 8mm,
  every node/.style={font=\small},
  lbl/.style={font=\scriptsize,midway},
  box/.style={draw,rounded corners,align=center,minimum height=9mm,text width=26mm,fill=stqpgray},
  direct/.style={draw=stqpgreen,very thick,rounded corners,align=center,minimum height=9mm,text width=26mm,fill=green!5},
  generic/.style={draw=stqporange,very thick,rounded corners,align=center,minimum height=9mm,text width=26mm,fill=orange!6},
  arr/.style={-{Latex[length=2.3mm]},thick}
]
\node[box,align=center] (problem) {Discrete\\problem};
\node[direct,right=of problem] (compat) {Compatibility or\\selection graph};
\node[direct,right=of compat] (clique) {Clique or\\weighted clique};
\node[box,right=of clique] (stqp) {$\displaystyle\min_{x\in\Delta}x^{\mathsf T}Qx$};
\node[generic,below=14mm of compat] (bilp) {Binary linear\\feasibility model};
\node[generic,right=of bilp] (qubo) {Factor model\\of any arity};
\draw[arr] (problem)--(compat);
\draw[arr] (compat)--(clique);
\draw[arr] (clique)--(stqp);
\draw[arr] (problem.south)|-(bilp.west);
\draw[arr] (bilp)--(qubo);
\draw[arr] (qubo)--node[lbl,right]{factor states}(clique);
\end{tikzpicture}
\caption{Compilation routes. Direct compatibility and selection graphs (green) preserve the problem's structure and yield the simplest StQP matrices: three-level for clique, $\{0,1\}$ off-diagonal for weighted clique. The generic routes (orange) reach StQP through the same weighted-clique step. The binary-linear route is polynomial in input bit length; the factor-model route is polynomial in the total explicit local-table encoding length.}
\label{fig:pipeline}
\end{figure}

\section{A three-level graph-to-StQP compiler}
\label{sec:graphcompiler}

Let $G=(V,E)$ be a simple undirected graph with $V=[n]$, adjacency matrix $A_G$, and clique number $\omega(G)$. For a nonempty clique $C\subseteq V$, define its barycenter
\begin{equation}
  x(C):=\frac{1}{|C|}\sum_{i\in C}e_i\in\Delta_n.
  \label{eq:cliquebary}
\end{equation}
The classical Motzkin--Straus theorem expresses $\omega(G)$ through quadratic optimization over $\Delta_n$ \citep{MotzkinStraus1965}. Regularized formulations sharpen the relationship between local optima and maximal cliques \citep{Bomze1997,GibbonsEtAl1997,HungerfordRinaldi2019}. We state the result in a minimization form designed for direct StQP compilation. \Cref{alg:graphcompiler} in \Cref{app:pseudocode} gives the compilation steps, and \Cref{app:formulations} works out several catalog entries explicitly.

\subsection{The graph matrix}

For $0<\tau<1$, define
\begin{equation}
  Q_\tau(G):=\one\one^{\mathsf T}-A_G-(1-\tau)I.
  \label{eq:Qtau}
\end{equation}
Its entries are
\begin{equation}
  (Q_\tau)_{ij}=
  \begin{cases}
    \tau, & i=j,\\
    0, & i\ne j\text{ and }\{i,j\}\in E,\\
    1, & i\ne j\text{ and }\{i,j\}\notin E.
  \end{cases}
  \label{eq:threelevels}
\end{equation}
Thus graph edges represent compatible co-activation and nonedges represent repulsive pair interactions.

A clique is a set of pairwise adjacent vertices. It is maximal if no additional vertex can be added while preserving the clique property, and maximum if its cardinality is largest among all cliques of the graph.

\begin{theorem}[Three-level graph compiler]
\label{thm:graphcompiler}
Let $G$ be a nonempty graph and $0<\tau<1$. Then
\begin{equation}
  \nu(Q_\tau(G))=\frac{\tau}{\omega(G)}.
  \label{eq:graphvalue}
\end{equation}
Moreover:
\begin{enumerate}[label=(\roman*)]
\item $x\in\Delta_n$ is a global minimizer if and only if $x=x(C)$ for a maximum clique $C$;
\item $x\in\Delta_n$ is a local minimizer if and only if $x=x(C)$ for a maximal clique $C$;
\item every local minimizer is strict and isolated.
\end{enumerate}
\end{theorem}

\begin{remark}[Relation to regularized Motzkin--Straus]
\label{rem:msattribution}
\Cref{thm:graphcompiler} is a minimization-form restatement of the regularized Motzkin--Straus theorem rather than a new result. Since $\one^{\mathsf T}x=1$ on $\Delta_n$,
\begin{equation}
  x^{\mathsf T}Q_\tau(G)\,x
  =1-x^{\mathsf T}\!\left(A_G+(1-\tau)I\right)x,
  \label{eq:msidentity}
\end{equation}
so minimizing $Q_\tau(G)$ over $\Delta_n$ is equivalent to maximizing the regularized Motzkin--Straus quadratic form \citep{Bomze1997}, whose $\tau=1/2$ case gives $\max_{\Delta_n}x^{\mathsf T}(A_G+\tfrac12 I)x=1-1/(2\omega(G))$ together with the bijection between strict local maximizers and maximal cliques. The $\tau$-parametrized family is the regularization studied in \citep{HungerfordRinaldi2019}, and $\Qhat(G)=2Q_{1/2}(G)$ is a positive rescaling of the same object. We give the proof because the minimization normalization, the nonnegative three-level alphabet, and the explicit gradient constants $2\tau/k$ and $2r_j(C)/k$ are used repeatedly in \Cref{sec:csp,sec:weightedcompiler,sec:packing}.
\end{remark}

\begin{proof}
We divide the argument into support purification, minimization on a clique, and local optimality.

Suppose $x\in\Delta_n$ has positive coordinates $x_i,x_j>0$ for a nonedge $\{i,j\}\notin E$. Along
\[
  x(t)=x+t(e_i-e_j),\qquad -x_i\le t\le x_j,
\]
the second derivative of $f(t)=x(t)^{\mathsf T}Q_\tau x(t)$ is
\[
  f''(t)=2(Q_{ii}+Q_{jj}-2Q_{ij})=4(\tau-1)<0.
\]
Hence $f$ is strictly concave on this segment. An interior point cannot minimize it locally, and at least one endpoint has strictly smaller objective. Repeating this elimination shows that every global minimizer has clique support, and any feasible point can be transformed without increasing the objective into a point supported on a clique.

If $x$ is supported on a clique $C$, all off-diagonal terms within $C$ vanish, so
\[
  x^{\mathsf T}Q_\tau x=\tau\sum_{i\in C}x_i^2
  \ge \frac{\tau}{|C|},
\]
with equality exactly when $x_i=1/|C|$ for $i\in C$. Since $|C|\le\omega(G)$, the global lower bound is $\tau/\omega(G)$, attained exactly at barycenters of maximum cliques.

It remains to characterize local minima. The strict-concavity argument shows that local minima must have clique support. On the face generated by a clique $C$, the unique face minimizer is $x(C)$ and the reduced Hessian is positive definite. Let $k=|C|$. At $x(C)$, the gradient components are
\[
  \nabla_i f=\frac{2\tau}{k}\quad(i\in C),
  \qquad
  \nabla_j f=\frac{2r_j(C)}{k}\quad(j\notin C),
\]
where $r_j(C)$ is the number of vertices in $C$ not adjacent to $j$. If $C$ is maximal, then $r_j(C)\ge1$ for every $j\notin C$, and since $\tau<1$,
\[
  \nabla_j f>\frac{2\tau}{k}.
\]
Thus the Karush–Kuhn–Tucker (KKT) inequalities hold with strict complementarity, and $x(C)$ is a strict local minimizer. Conversely, if $C$ is not maximal, some $j\notin C$ is adjacent to every vertex of $C$, so $r_j(C)=0$ and moving a small amount of mass toward $e_j$ decreases the objective. Hence $x(C)$ is not locally minimizing. Finiteness of the maximal-clique family gives isolation.
\end{proof}

\begin{remark}[Why regularization matters]
At $\tau=1$, the matrix is $\one\one^{\mathsf T}-A_G$, and the global value reduces to the classical Motzkin--Straus formula. The strict support-purification and the exact local-minimum correspondence can fail in this unregularized case. The range $0<\tau<1$ removes nonclique local minima and makes every maximal clique a strict local minimum.
\end{remark}

\subsection{An integer coefficient alphabet}

Set $\tau=1/2$ and scale by two:
\begin{equation}
  \Qhat(G):=2Q_{1/2}(G)=2\one\one^{\mathsf T}-2A_G-I.
  \label{eq:Qhat}
\end{equation}
Then
\begin{equation}
  \widehat q_{ij}=
  \begin{cases}
    1, & i=j,\\
    0, & \{i,j\}\in E,\\
    2, & i\ne j\text{ and }\{i,j\}\notin E,
  \end{cases}
  \qquad
  \nu(\Qhat(G))=\frac{1}{\omega(G)}.
  \label{eq:integergraph}
\end{equation}
This formulation uses no negative coefficients and only the alphabet $\{0,1,2\}$. The discrete solution is read from the support; the positive coordinates all equal $1/\omega(G)$.

\begin{corollary}[Restricted StQP threshold is NP-complete]
\label{cor:npc}
Consider the following problem.

\smallskip
\noindent\textnormal{\textsc{Three-Level StQP Threshold}}: Given a symmetric matrix $Q\in\{0,1,2\}^{n\times n}$ satisfying $q_{ii}=1$ and $q_{ij}\in\{0,2\}$ for $i\ne j$, and an integer $k\ge1$, decide whether
\[
  \min_{x\in\Delta_n}x^{\mathsf T}Qx\le\frac1k.
\]
This problem is NP-complete.
\end{corollary}

\begin{proof}
Given a graph $G$, construct $\Qhat(G)$ in polynomial time. By \Cref{thm:graphcompiler},
\[
  \nu(\Qhat(G))\le\frac1k
  \quad\Longleftrightarrow\quad
  \omega(G)\ge k.
\]
This proves NP-hardness by CLIQUE. Membership in NP follows directly from the matrix pattern: a certificate is a set $C$ of $k$ indices with $q_{ij}=0$ for all distinct $i,j\in C$. Such a set is a $k$-clique in the graph defined by the zero off-diagonal entries and can be checked in polynomial time.
\end{proof}

\begin{corollary}[Universality for NP]
\label{cor:universalNP}
Every decision language in NP has a polynomial-time many-one reduction to \textnormal{\textsc{Three-Level StQP Threshold}}.
\end{corollary}

\begin{proof}
By the Cook--Levin theorem and Karp reductions, every language in NP reduces in polynomial time to CLIQUE \citep{Cook1971,Karp1972}. Apply \Cref{cor:npc}.
\end{proof}

\begin{caution}
\Cref{cor:universalNP} is a representability statement, not an algorithmic speedup. The encoded matrices form a restricted NP-complete family, so a general exact polynomial-time StQP solver would imply $\mathrm P=\mathrm{NP}$.
\end{caution}

\subsection{Objective separation and perturbation tolerance}
\label{subsec:precision}

Compact coefficients do not imply a constant decision gap. If the relevant distinction is between clique number $k$ and $k-1$, with $k\ge2$, then for $Q_\tau$ the adjacent optimal values differ by
\begin{equation}
  g_k(\tau)
  :=\frac{\tau}{k-1}-\frac{\tau}{k}
  =\frac{\tau}{k(k-1)}.
  \label{eq:cliquegap}
\end{equation}
For the integer-scaled matrix $\Qhat$, the separation is $1/[k(k-1)]$.

\begin{proposition}[Entrywise perturbation bound]
\label{prop:perturb}
Let $Q,E\in\mathbb S^N$ and $\norm{E}_{\max}:=\max_{u,v}|E_{uv}|\le\varepsilon$. Then
\begin{equation}
  |\nu(Q+E)-\nu(Q)|\le\varepsilon.
  \label{eq:perturb}
\end{equation}
\end{proposition}

\begin{proof}
For every $x\in\Delta_N$,
\[
  |x^{\mathsf T}Ex|
  \le \sum_{u,v}|E_{uv}|x_ux_v
  \le\varepsilon\left(\sum_u x_u\right)^2
  =\varepsilon.
\]
Apply this inequality at minimizers of $Q$ and $Q+E$.
\end{proof}

Thus a midpoint threshold between the $k$ and $k-1$ value levels remains valid under arbitrary entrywise perturbations smaller than $g_k(\tau)/2$. The required objective discrimination decays as $O(k^{-2})$. Support recovery has a related scale: the ideal nonzero coordinate of a $k$-clique barycenter is $1/k$, so coordinate noise substantially below $1/(2k)$ permits simple threshold decoding.

\subsection{Landscape interpretation}
\label{subsec:landscape}

The graph compiler converts combinatorial maximality into continuous local optimality. If $\mathcal M(G)$ denotes the family of maximal cliques, then
\begin{equation}
  \#\{\text{strict local minima of }\StQP(Q_\tau(G))\}
  =|\mathcal M(G)|.
  \label{eq:numberlocal}
\end{equation}
Graphs can have exponentially many maximal cliques. Hence a simplex-native local relaxer can face an exponentially multimodal landscape even though the coefficient alphabet is fixed and nonnegative.

This observation cuts both ways. The landscape is undesirable for exact optimization but useful as a transparent diagnostic. A returned local minimum is not arbitrary: its support is a maximal clique, and its objective immediately reports that clique's size. Global exploration mechanisms can therefore be studied in a controlled setting where all local solutions have clear combinatorial meaning. Maximum-clique algorithms based on continuous dynamics exploit precisely this relationship \citep{Bomze1997,BomzePelilloStix2000,WuHao2015}.

\subsection{Relative couplings}
\label{subsec:relative}

The literal expression $Q_\tau=\one\one^{\mathsf T}-A_G-(1-\tau)I$ contains a dense all-ones term. On the simplex, however, $x^{\mathsf T}\one\one^{\mathsf T}x=1$ is constant. Therefore the same minimizers are obtained from
\begin{equation}
  \widetilde Q_\tau(G):=-A_G-(1-\tau)I,
  \label{eq:relativeQ}
\end{equation}
with all objective values shifted by $-1$. A physical platform that supplies a global normalization constraint need only realize the relative edge and diagonal interactions. Conversely, the nonnegative form \eqref{eq:threelevels} may be preferable when repulsive incompatibility couplings are natural. The two forms are mathematically equivalent but can have very different implementation costs.

\section{Compatibility graphs for constraint satisfaction and assignment}
\label{sec:csp}

Many discrete problems ask for one local choice from each of several groups, subject to pairwise consistency. Such problems have a natural compatibility graph: vertices represent local choices, and edges join choices that can coexist. The graph-to-StQP theorem then converts consistency into a simplex quadratic.

Choices belonging to the same group are represented by mutually nonadjacent vertices. A clique can therefore contain at most one choice from each group. When there are $n$ groups, a clique of size $n$ selects exactly one choice from every group, while adjacency between groups enforces consistency. Smaller cliques represent partial selections.

\subsection{Finite-domain constraint satisfaction}

A binary constraint satisfaction problem (CSP) consists of variables $z_1,\ldots,z_n$ with finite domains $D_1,\ldots,D_n$ and binary constraints $R_{ij}\subseteq D_i\times D_j$. Here binary refers to constraints involving two variables; the domains need not be Boolean. Unary restrictions may be absorbed into the domains. The microstructure graph is classical in constraint processing \citep{Jegou1993,Dechter2003,RossiVanBeekWalsh2006}.

\begin{construction}[Variable--value compatibility graph]
\label{con:cspgraph}
Create one vertex $(i,a)$ for every $i\in[n]$ and $a\in D_i$. Two distinct vertices $(i,a)$ and $(j,b)$ are adjacent if and only if $i\ne j$ and either no constraint links $i$ and $j$, or $(a,b)\in R_{ij}$. Vertices belonging to the same variable are nonadjacent.
\end{construction}

\begin{proposition}[Binary CSP as StQP]
\label{prop:csp}
Let $H$ be the graph from \Cref{con:cspgraph}. Then the CSP is satisfiable if and only if $H$ contains a clique of size $n$. Equivalently,
\begin{equation}
  \min_{x\in\Delta_N}x^{\mathsf T}\Qhat(H)x=\frac1n,
  \qquad N=\sum_{i=1}^n|D_i|.
  \label{eq:cspstqp}
\end{equation}
When satisfiable, the support of a global minimizer contains one variable--value vertex for each $i$ and decodes directly to a satisfying assignment.
\end{proposition}

\begin{proof}
A clique contains at most one vertex from each variable group. A clique of size $n$ therefore selects exactly one value per variable, and adjacency enforces all binary constraints. Conversely, a satisfying assignment selects $n$ pairwise adjacent vertices. Since no clique can exceed $n$, \Cref{thm:graphcompiler} gives \eqref{eq:cspstqp}.
\end{proof}

The selected variable--value pairs are identified by the support of the StQP solution. The coordinates themselves are continuous masses, with one global normalization; the size-$n$ clique condition supplies the one-choice-per-variable requirement.

\subsection{CNF satisfiability}

Continuous formulations of satisfiability have a long history, from global-optimization models of SAT \citep{gu-sat-global-optimization-1994} to recent Fourier and polynomial relaxations \citep{kyrillidis-fouriersat-2020,fouriercsp-2025}; the construction below is exact and keeps the three-level alphabet. Let a conjunctive normal form (CNF) formula $\Phi=C_1\wedge\cdots\wedge C_m$ contain $N_{\mathrm{lit}}$ literal occurrences. A particularly compact construction uses one vertex per occurrence, rather than two vertices per Boolean variable.

\begin{construction}[Clause--literal graph]
\label{con:satgraph}
Create a vertex $(r,\ell)$ for every literal occurrence $\ell\in C_r$. Join $(r,\ell)$ and $(s,\ell')$ when $r\ne s$ and the two literals are not complementary. Vertices from the same clause are nonadjacent.
\end{construction}

\begin{proposition}[CNF-SAT]
\label{prop:sat}
The formula $\Phi$ is satisfiable if and only if the graph in \Cref{con:satgraph} has a clique of size $m$. Consequently,
\begin{equation}
  \Phi\text{ satisfiable}
  \quad\Longleftrightarrow\quad
  \nu(\Qhat(H_\Phi))=\frac1m.
  \label{eq:satstqp}
\end{equation}
The formulation has $N_\Delta=N_{\mathrm{lit}}$ coordinates; for 3-SAT, $N_{\mathrm{lit}}\le3m$.
\end{proposition}

\begin{proof}
A satisfying assignment makes at least one literal true in every clause. Selecting one true occurrence per clause gives pairwise noncomplementary literals and hence an $m$-clique. Conversely, an $m$-clique contains one occurrence from every clause and contains no complementary pair. Assigning all selected positive literals true and all selected negative literals false is consistent and can be extended arbitrarily to unselected variables, satisfying every clause.
\end{proof}

A local minimum of the associated StQP corresponds to a maximal set of mutually consistent clause literals. Such a support may omit clauses and therefore need not certify satisfiability. This gives a concrete example of why local minimization alone is insufficient.

\begin{example}[A 3-SAT decoding rule]
Suppose a global minimizer of $\Qhat(H_\Phi)$ has support
\[
  \set{(1,x_2),(2,\neg x_5),(3,x_2),\ldots,(m,x_7)}.
\]
All nonzero coordinates equal $1/m$. The support cannot contain both $x_i$ and $\neg x_i$. Set each selected positive literal to true and each selected negative literal to false; assign all remaining variables arbitrarily. The result satisfies $\Phi$.
\end{example}

\subsection{Graph coloring and clique cover}

For a graph $G=(V,E)$ and integer $k$, create vertices $(v,c)$ with $v\in V$ and $c\in[k]$. Two assignments $(v,c)$ and $(u,d)$ are compatible when $v\ne u$ and either $\{u,v\}\notin E$ or $c\ne d$.

\begin{proposition}[$k$-colorability]
\label{prop:coloring}
Let $H_{G,k}$ be the compatibility graph above. Then
\begin{equation}
  \chi(G)\le k
  \quad\Longleftrightarrow\quad
  \omega(H_{G,k})=|V|
  \quad\Longleftrightarrow\quad
  \nu(\Qhat(H_{G,k}))=\frac1{|V|}.
  \label{eq:colorstqp}
\end{equation}
The StQP uses $N_\Delta=k|V|$ coordinates, and a global support selects one color for every vertex.
\end{proposition}

Covering $G$ by at most $k$ cliques is equivalent to coloring the complement graph $\overline G$ with at most $k$ colors. Hence the same construction yields a clique-cover formulation with $k|V|$ coordinates.

\subsection{Hamiltonian cycles}

Let $G=(V,E)$ have $n\ge2$ vertices in the directed case and $n\ge3$ in the simple undirected case; smaller instances are handled separately. Use vertices $(p,v)$ meaning that graph vertex $v$ occupies tour position $p\in\{0,\ldots,n-1\}$. Two such assignments are compatible if they use different positions and different graph vertices, and if assignments to consecutive positions follow an edge of $G$.

For a directed graph, more explicitly, $(p,u)$ and $(q,v)$ are adjacent when:
\begin{enumerate}[label=(\roman*)]
\item $p\ne q$ and $u\ne v$;
\item if $q\equiv p+1\pmod n$, then $(u,v)\in E$;
\item if $p\equiv q+1\pmod n$, then $(v,u)\in E$.
\end{enumerate}
For nonconsecutive positions, no graph-edge condition is imposed.

\begin{proposition}[Hamiltonian cycle]
\label{prop:hamiltonian}
The directed graph $G$ contains a Hamiltonian cycle if and only if the compatibility graph contains a clique of size $n$. Therefore
\begin{equation}
  G\text{ Hamiltonian}
  \quad\Longleftrightarrow\quad
  \nu(\Qhat(H_G^{\mathrm{Ham}}))=\frac1n.
  \label{eq:hamstqp}
\end{equation}
The formulation has $N_\Delta=n^2$ coordinates. The undirected case uses the same construction with undirected adjacency.
\end{proposition}

\begin{proof}
An $n$-clique contains at most one vertex for each position and at most one occurrence of each graph vertex, so it defines a permutation of $V$. Compatibility at consecutive positions makes that permutation a Hamiltonian cycle. The converse follows by selecting the $n$ position--vertex assignments from any Hamiltonian cycle.
\end{proof}

The construction is closely related to one-hot QUBO encodings of tours, but its constraint logic is expressed solely through the compatibility graph. There is no separate penalty-scale choice: all conflicts receive the same coefficient in $\Qhat$.

\subsection{Graph isomorphism and association graphs}

Graph isomorphism is not known to be NP-complete, but it is a useful illustration of the association-graph principle and has a well-established continuous clique interpretation \citep{Pelillo1999}; association graphs underlie continuous methods for matching, common subgraphs, and shape correspondence \citep{pelillo-siddiqi-zucker-assoc-1999,leordeanu-hebert-spectral-2005}, and clique and constraint models are established approaches to maximum common subgraph \citep{mccreesh-clique-constraint-mcs-2016}. Let $G=(V_G,E_G)$ and $H=(V_H,E_H)$ have $n$ vertices. Create one vertex $(u,v)$ for each possible mapping $u\in V_G$ to $v\in V_H$. Two mapping vertices $(u,v)$ and $(u',v')$ are adjacent when $u\ne u'$, $v\ne v'$, and
\begin{equation}
  \{u,u'\}\in E_G
  \quad\Longleftrightarrow\quad
  \{v,v'\}\in E_H.
  \label{eq:isocompat}
\end{equation}

\begin{proposition}[Graph isomorphism]
The graphs $G$ and $H$ are isomorphic if and only if the association graph has a clique of size $n$, equivalently if the associated three-level StQP has optimum $1/n$. The formulation uses $n^2$ coordinates.
\end{proposition}

For subgraph isomorphism from $G$ into $H$, allow different graph orders, replace the equivalence in \eqref{eq:isocompat} by the implication $\{u,u'\}\in E_G\Rightarrow\{v,v'\}\in E_H$, and use $|V_G|$ as the target clique size. Keeping the equivalence while allowing partial injective supports instead gives the maximum common \emph{induced} subgraph formulation. Weighted vertex scores can be included through the weighted clique construction below.

\subsection{Arbitrary-arity hard constraints}
\label{subsec:arity}

For constraints involving more than two variables, use one vertex for every allowed tuple of every constraint. Vertices from different constraints are adjacent when the tuples agree on all shared variables; vertices from the same constraint are nonadjacent. A clique containing one tuple from every constraint is equivalent to a globally consistent satisfying assignment. If there are $F$ constraints and $A_\alpha$ allowed tuples for constraint $\alpha$, then
\begin{equation}
  N_\Delta=\sum_{\alpha=1}^F|A_\alpha|.
  \label{eq:aritysize}
\end{equation}
This construction is polynomial in the explicit allowed-tuple list, but expanding succinct high-arity constraints can require exponentially many tuples. The tuple graph is the feasibility counterpart of the weighted factor-state compiler developed next. The two are complementary rather than nested: pruning to allowed tuples is exact here, but falls outside the hypotheses of the weighted compiler (\Cref{rem:nopruning}).

\section{Weighted clique and a generic factor-state compiler}
\label{sec:weightedcompiler}

The unweighted graph compiler captures feasibility and cardinality. Optimization objectives require weights. We first prove a weighted clique StQP formula, then use it to compile arbitrary finite-domain factor objectives.

\subsection{A weighted clique StQP theorem}

Let $G=(V,E)$ have positive vertex weights $w_i>0$. For a clique $C$, write
\[
  w(C):=\sum_{i\in C}w_i,
  \qquad
  \omega_w(G):=\max\set{w(C):C\text{ is a clique}}.
\]
Let $w_{\min}:=\min_i w_i$, $w_{\max}:=\max_i w_i$, and choose $0<\tau<1$. Define $Q_{\tau,w}(G)$ entrywise by
\begin{equation}
  q_{ij}=
  \begin{cases}
    \displaystyle \tau\frac{w_{\min}}{w_i}, & i=j,\\[5pt]
    0, & i\ne j\text{ and }\{i,j\}\in E,\\
    1, & i\ne j\text{ and }\{i,j\}\notin E.
  \end{cases}
  \label{eq:weightedQ}
\end{equation}
This is a regularized version of weighted Motzkin--Straus formulations \citep{BomzePelilloStix2000,Busygin2006}.

\begin{theorem}[Weighted graph compiler]
\label{thm:weightedclique}
For the matrix in \eqref{eq:weightedQ},
\begin{equation}
  \nu(Q_{\tau,w}(G))
  =\frac{\tau w_{\min}}{\omega_w(G)}.
  \label{eq:weightedvalue}
\end{equation}
There exists a global minimizer supported on a maximum-weight clique $C^\star$, with
\begin{equation}
  x_i^\star=
  \begin{cases}
    \displaystyle\frac{w_i}{w(C^\star)}, & i\in C^\star,\\[5pt]
    0, & i\notin C^\star.
  \end{cases}
  \label{eq:weightedbary}
\end{equation}
Every local minimizer, and therefore every global minimizer, has clique support. An arbitrary feasible point can be purified to clique support without increasing the objective.
\end{theorem}

\begin{proof}
If positive coordinates $x_i,x_j$ correspond to a nonedge, then along $x+t(e_i-e_j)$ the second derivative is
\[
  2\left(\tau\frac{w_{\min}}{w_i}
  +\tau\frac{w_{\min}}{w_j}-2\right)
  \le4(\tau-1)<0.
\]
Thus a local minimizer cannot contain a nonedge in its support, and an arbitrary feasible point can be purified by endpoint moves until its support is a clique.

For a fixed clique $C$, off-diagonal terms vanish and
\[
  x^{\mathsf T}Q_{\tau,w}x
  =\tau w_{\min}\sum_{i\in C}\frac{x_i^2}{w_i}.
\]
Cauchy--Schwarz gives
\[
  1=\left(\sum_{i\in C}x_i\right)^2
  \le\left(\sum_{i\in C}\frac{x_i^2}{w_i}\right)
      \left(\sum_{i\in C}w_i\right),
\]
so the face minimum is $\tau w_{\min}/w(C)$, attained exactly at \eqref{eq:weightedbary}. Minimizing over cliques selects a maximum-weight clique and proves \eqref{eq:weightedvalue}.
\end{proof}

\begin{remark}
\label{rem:weightscaling}
The diagonal values encode the vertex weights through reciprocals, while incompatibility remains binary. Scaling all $w_i$ by a common positive constant changes neither the matrix nor the solution. The coefficient dynamic range on the diagonal is $w_{\max}/w_{\min}$.
\end{remark}

\Cref{thm:weightedclique} locates the global minimum. The next theorem describes every local minimizer, extending \Cref{thm:graphcompiler} to positive weights.

\begin{theorem}[Weighted landscape]
\label{thm:weightedlandscape}
Let $G$ be a nonempty graph with positive vertex weights $w$, let $0<\tau<1$, and let $Q_{\tau,w}(G)$ be as in \eqref{eq:weightedQ}. For a clique $C$ write $x(C,w):=\sum_{i\in C}(w_i/w(C))\,e_i\in\Delta_n$ for its weighted barycenter; \eqref{eq:weightedbary} is $x(C^\star,w)$. Then:
\begin{enumerate}[label=(\roman*)]
\item $x\in\Delta_n$ is a local minimizer if and only if $x=x(C,w)$ for a maximal clique $C$;
\item every local minimizer is strict and isolated;
\item $x$ is a global minimizer if and only if $x=x(C,w)$ for a maximum-weight clique.
\end{enumerate}
In particular the number of strict local minimizers equals $|\mathcal M(G)|$, the number of maximal cliques, independently of the weights.
\end{theorem}

\begin{proof}
By the purification argument in the proof of \Cref{thm:weightedclique}, every local minimizer has clique support. If $x$ is a local minimizer with support exactly the clique $C$, then $x$ lies in the relative interior of the face generated by $C$ and minimizes $f=x^{\mathsf T}Q_{\tau,w}x$ locally there; on that face $f=\tau w_{\min}\sum_{i\in C}x_i^2/w_i$ is strictly convex, so $x=x(C,w)$ by \Cref{thm:weightedclique}. Local minimizers are therefore among the weighted barycenters of cliques.

Fix a clique $C$ and put $x^\star=x(C,w)$. For $j\notin C$ let
\begin{equation}
  r^w_j(C):=\sum_{\substack{i\in C\\ \{i,j\}\notin E}} w_i
  \label{eq:weighteddeficiency}
\end{equation}
be the total weight of vertices of $C$ not adjacent to $j$. Since off-diagonal entries vanish inside $C$,
\begin{equation}
  (Q_{\tau,w}x^\star)_i=\tau\frac{w_{\min}}{w_i}\cdot\frac{w_i}{w(C)}
  =\frac{\tau w_{\min}}{w(C)}
  \qquad(i\in C),
  \label{eq:wgradin}
\end{equation}
which is constant on $C$, while for $j\notin C$ only nonedges to $C$ contribute, each with coefficient one, so
\begin{equation}
  (Q_{\tau,w}x^\star)_j=\frac{r^w_j(C)}{w(C)}.
  \label{eq:wgradout}
\end{equation}

Suppose $C$ is maximal. Every $j\notin C$ then has a non-neighbour in $C$, so $r^w_j(C)\ge w_{\min}>\tau w_{\min}$ because $\tau<1$; comparing \eqref{eq:wgradin} and \eqref{eq:wgradout}, the KKT conditions hold with strict complementarity. For strictness, let $x\in\Delta_n$ with $d:=x-x^\star\ne0$. Then $\one^{\mathsf T}d=0$, and $d_j=x_j\ge0$ for $j\notin C$ because $x^\star_j=0$ there; put $s:=\sum_{j\notin C}d_j\ge0$. Using $\sum_{i\in C}d_i=-s$,
\begin{equation}
  d^{\mathsf T}Q_{\tau,w}x^\star
  =\frac{1}{w(C)}\sum_{j\notin C}d_j\bigl(r^w_j(C)-\tau w_{\min}\bigr)
  \ \ge\ \frac{(1-\tau)w_{\min}}{w(C)}\,s
  \ \ge\ 0 .
  \label{eq:wfirstorder}
\end{equation}
Split $d=u+v$ with $u$ supported on $C$ and $v\ge0$ supported off $C$, so $\norm{v}_1=s$ and $\norm{u}_1\le\norm{d}_1$. The matrix $Q_{\tau,w}$ is entrywise nonnegative with entries at most one and is diagonal on $C$, so $u^{\mathsf T}Q_{\tau,w}u\ge0$, $v^{\mathsf T}Q_{\tau,w}v\ge0$ and $\abs{u^{\mathsf T}Q_{\tau,w}v}\le\norm{u}_1 s$. Hence
\[
  f(x)-f(x^\star)=2d^{\mathsf T}Q_{\tau,w}x^\star+d^{\mathsf T}Q_{\tau,w}d
  \ \ge\ 2s\left(\frac{(1-\tau)w_{\min}}{w(C)}-\norm{d}_1\right).
\]
If $s>0$ and $\norm{d}_1<(1-\tau)w_{\min}/w(C)$, the right-hand side is positive. If $s=0$, then $d=u\ne0$, the linear term vanishes, and $f(x)-f(x^\star)=\tau w_{\min}\sum_{i\in C}u_i^2/w_i>0$. Thus $f(x)>f(x^\star)$ for every $x\in\Delta_n$ with $0<\norm{x-x^\star}_1<(1-\tau)w_{\min}/w(C)$: the minimizer is strict.

If $C$ is not maximal, some $j\notin C$ is adjacent to every vertex of $C$, so $r^w_j(C)=0$ and $(e_j-x^\star)^{\mathsf T}Q_{\tau,w}x^\star=-\tau w_{\min}/w(C)<0$; moving mass toward $e_j$ decreases the objective and $x(C,w)$ is not locally minimizing. This proves (i) and strictness. Since the graph has finitely many maximal cliques, (i) also gives isolation, completing (ii). For (iii), a global minimizer exists by compactness and is a local minimizer, so by (i) it equals $x(C,w)$ for a maximal clique $C$, with value $\tau w_{\min}/w(C)$. This value is minimal exactly when $w(C)=\omega_w(G)$, and every maximum-weight clique is maximal because the weights are positive.
\end{proof}

\paragraph{Relation to the weighted regularization of Bomze--Pelillo--Stix.}
The local/global clique correspondence and strictness are already established for a class of weighted regularizations in \citet{BomzePelilloStix2000}. In their normalization, a matrix $C$ has $C_{ii}=1/(2w_i)$, zero edge entries, and $C_{ij}\ge C_{ii}+C_{jj}$ on nonedges. The rescaling $C=Q_{\tau,w}/(2\tau w_{\min})$ satisfies this condition when
\[
  \tau w_{\min}\left(\frac1{w_i}+\frac1{w_j}\right)\le1
  \qquad\text{for every nonedge }\{i,j\},
\]
in particular for $0<\tau\le1/2$. The proof above verifies the binary off-diagonal normalization for the full range $0<\tau<1$, without requiring that stronger matrix-class condition. We use its explicit constants in the factor-state and precision analyses; the weighted clique correspondence itself is not a new principle.

\begin{remark}[The landscape is weight-independent]
\label{rem:weightindep}
Weights move the \emph{values} of the local minima but not their \emph{number} or their combinatorial meaning: the strict local minimizers of $\StQP(Q_{\tau,w}(G))$ stand in bijection with $\mathcal M(G)$ for every positive weight vector, exactly as in the unweighted case \eqref{eq:numberlocal}. The consequences for complete factor-state models are given in \Cref{cor:allassignments}.
\end{remark}

\subsection{Finite-domain factor models}

Let $z=(z_1,\ldots,z_n)$, where $z_i$ belongs to a finite domain $D_i$. Consider
\begin{equation}
  E(z)=\sum_{\alpha\in\mathcal F}\phi_\alpha(z_{S_\alpha}),
  \label{eq:factorobjective}
\end{equation}
where $S_\alpha\subseteq[n]$ is the scope of factor $\alpha$, and $z_{S_\alpha}$ denotes the local assignment on that scope. No sign assumption is imposed on $\phi_\alpha$. The identities below allow real table values. For polynomial-time compilation, the factors are supplied as explicit complete rational tables, with rational offsets of polynomial encoding length; a fixed rational $\tau$ may be used. Complexity is measured in the total table encoding length, not in a succinct description of arbitrary-arity factors. A single factor on $r$ binary variables has $2^r$ local states.

\begin{construction}[Factor-state compatibility graph]
\label{con:factorgraph}
For each factor $\alpha$ and each local state
\[
  a\in\mathcal A_\alpha:=\prod_{i\in S_\alpha}D_i,
\]
create a vertex $v=(\alpha,a)$. Vertices belonging to the same factor are nonadjacent. Vertices $(\alpha,a)$ and $(\beta,b)$ with $\alpha\ne\beta$ are adjacent when $a$ and $b$ agree on every variable in $S_\alpha\cap S_\beta$.

Choose constants
\begin{equation}
  B_\alpha>\max_{a\in\mathcal A_\alpha}\phi_\alpha(a)
  \label{eq:Balpha}
\end{equation}
and give vertex $(\alpha,a)$ the positive weight
\begin{equation}
  w_{\alpha,a}:=B_\alpha-\phi_\alpha(a).
  \label{eq:factorweight}
\end{equation}
Write $B:=\sum_\alpha B_\alpha$.
\end{construction}

\begin{theorem}[Exact factor-state compiler]
\label{thm:factorcompiler}
Assume every variable that matters to the objective appears in at least one factor, and that $H_E$ carries a vertex for \emph{every} local state $a\in\mathcal A_\alpha=\prod_{i\in S_\alpha}D_i$ of every factor, with no states pruned. Let $H_E$ be the weighted compatibility graph in \Cref{con:factorgraph}. Then
\begin{equation}
  \omega_w(H_E)=B-E^\star,
  \qquad
  E^\star:=\min_z E(z).
  \label{eq:factorcliquevalue}
\end{equation}
Consequently, for the weighted StQP matrix $Q_E=Q_{\tau,w}(H_E)$,
\begin{equation}
  \nu(Q_E)=\frac{\tau w_{\min}}{B-E^\star},
  \qquad
  E^\star=B-\frac{\tau w_{\min}}{\nu(Q_E)}.
  \label{eq:factorrecover}
\end{equation}
Every global minimizer decodes to a globally optimal assignment. The simplex dimension is
\begin{equation}
  N_\Delta=\sum_{\alpha\in\mathcal F}|\mathcal A_\alpha|.
  \label{eq:factorsize}
\end{equation}
\end{theorem}

\begin{proof}
Any global assignment $z$ chooses one local state $(\alpha,z_{S_\alpha})$ from each factor. These vertices agree on overlaps and form a clique $C_z$ with weight
\[
  w(C_z)=\sum_\alpha\left(B_\alpha-\phi_\alpha(z_{S_\alpha})\right)
  =B-E(z).
\]
Thus $\omega_w(H_E)\ge B-E^\star$.

Conversely, vertices in a clique are pairwise consistent, so their local states combine into one partial global assignment. If a clique omits factor $\beta$, choose a local state for $\beta$ that agrees with this partial assignment on already assigned variables and assign arbitrary values to the remaining variables in $S_\beta$. Because all local states are present, this vertex exists and is adjacent to the clique. Its weight is positive, so no maximum-weight clique can omit a factor. Hence every maximum-weight clique contains exactly one state from every factor and corresponds to a global assignment $z$, with weight $B-E(z)$. This proves \eqref{eq:factorcliquevalue}. Equation \eqref{eq:factorrecover} follows from \Cref{thm:weightedclique}. The support states combine to the optimal assignment.
\end{proof}

\begin{remark}[Hard constraints must not prune states]
\label{rem:nopruning}
The completeness hypothesis on $\mathcal A_\alpha$ is essential to the stated theorem. The proof extends a clique that omits factor $\beta$ by exhibiting a state of $\beta$ consistent with the partial assignment already fixed; that state is guaranteed to exist precisely because $\mathcal A_\beta$ is the full product of domains. If disallowed states are deleted instead of penalized, the extension can fail: a maximum-weight clique may then legitimately omit a factor, and \eqref{eq:factorcliquevalue} is false. A minimal example has one variable $z_1\in\{0,1\}$ and two factors with scope $\{z_1\}$: $\phi_\alpha(0)=10$, $\phi_\alpha(1)=0$, $B_\alpha=11$, and $\beta$ enforces $z_1=0$ by keeping only the state $0$, with $\phi_\beta(0)=0$, $B_\beta=1$. The states $(\alpha,1)$ and $(\beta,0)$ disagree, so the cliques are $\{(\alpha,0),(\beta,0)\}$ of weight $2$ and $\{(\alpha,1)\}$ of weight $11$. Thus $\omega_w=11$, whereas $B-E^\star=2$, and the decoded assignment $z_1=1$ violates the constraint.

The proof uses less than completeness: it suffices that, for every factor $\beta$, every assignment to the variables that $S_\beta$ shares with other factors extends to a retained state of $\beta$. States that differ only on variables private to $S_\beta$ may therefore be pruned, and \eqref{eq:factorcliquevalue} then holds with $E^\star$ the minimum over retained assignments.

To retain the complete-state guarantee without a separate extendibility argument, encode hard constraints through the values $\phi_\alpha$ rather than deleting states. Write $\phi_\alpha=\psi_\alpha+P\,\mathbf 1[\text{constraint }\alpha\text{ violated}]$ with soft parts $\psi_\alpha$. If some assignment satisfies every constraint and
\[
  P>S_{\mathrm{soft}}:=\sum_{\gamma}\Bigl(\max_a\psi_\gamma(a)-\min_a\psi_\gamma(a)\Bigr),
\]
then every minimizer of $E$ is feasible, since a feasible $z$ has $E(z)\le\sum_\gamma\max\psi_\gamma$ and a violating one has $E(z)\ge P+\sum_\gamma\min\psi_\gamma$. \Cref{rem:offsets} discusses how larger penalties can affect coefficient ratios and conservative separation bounds. In particular the unweighted tuple graph of \Cref{subsec:arity}, which carries one vertex per \emph{allowed} tuple, is not a special case of \Cref{con:factorgraph}; it is the complementary construction, exact for pure feasibility but outside the scope of \Cref{thm:factorcompiler}.
\end{remark}

\begin{corollary}[Every assignment is a strict local minimizer]
\label{cor:allassignments}
Under the hypotheses of \Cref{thm:factorcompiler}, suppose every variable lies in the scope of some factor. Then the maximal cliques of $H_E$ are exactly the sets
\[
  C_z:=\set{(\alpha,z_{S_\alpha}):\alpha\in\mathcal F},
  \qquad z\in\textstyle\prod_i D_i .
\]
Consequently $\StQP(Q_E)$ has exactly $\prod_i|D_i|$ strict local minimizers, one for each assignment $z$, with value $\tau w_{\min}/(B-E(z))$. In particular the encoding in \Cref{cor:qubo} of an $n$-variable QUBO in which every variable appears in some term has $2^n$ strict local minimizers.
\end{corollary}

\begin{proof}
$C_z$ is a clique by the proof of \Cref{thm:factorcompiler}, and it is maximal because any other state $(\alpha,a)$ with $a\ne z_{S_\alpha}$ is nonadjacent to $(\alpha,z_{S_\alpha})$. Conversely, the extension step in the same proof shows that a clique omitting a factor is not maximal, and a clique with one state per factor is $C_z$ for a unique $z$ because every variable is covered. Apply \Cref{thm:weightedlandscape}(i)--(ii).
\end{proof}

\begin{remark}[What the landscape does not filter]
\label{rem:landscapefilter}
The continuous landscape of a compiled factor model eliminates no configuration: every assignment, feasible or not, optimal or not, is a strict local minimum, and only the ordering of the local values $\tau w_{\min}/(B-E(z))$ reflects the objective. This is a weaker filter than a QUBO's, whose local minima under single-bit flips are only the $1$-flip-stable assignments; near $x(C_z,w)$ there is no descent direction that changes even one variable. On the binary-linear route of \Cref{sec:bilp}, every infeasible configuration is likewise a strict local minimizer. Any advantage of a continuous method on these encodings must therefore come from its trajectory and basins of attraction from the interior of the simplex, not from an absence of spurious local minima.
\end{remark}

\begin{corollary}[Decision threshold]
For any $T<B$,
\begin{equation}
  E^\star\le T
  \quad\Longleftrightarrow\quad
  \nu(Q_E)\le\frac{\tau w_{\min}}{B-T}.
  \label{eq:factorthreshold}
\end{equation}
\end{corollary}

\subsection{Sparsity of the incompatibility matrix}

Although the compatibility graph can be dense, the off-diagonal nonzeros of $Q_E$ correspond to \emph{incompatibilities}. Factors with disjoint scopes are fully compatible and produce zero cross blocks. Nonzero interactions occur only:
\begin{enumerate}[label=(\roman*)]
\item within each factor block, to prevent selecting two states of the same factor;
\item between factor blocks that share a variable, and only for state pairs that disagree on that variable.
\end{enumerate}
A crude upper bound is
\begin{equation}
  M_Q=O\left(
  \sum_\alpha |\mathcal A_\alpha|^2
  +\sum_{i=1}^n
  \left(\sum_{\alpha:i\in S_\alpha}|\mathcal A_\alpha|\right)^2
  \right).
  \label{eq:factorcouplings}
\end{equation}
For uniformly bounded domain sizes, factor arity, and variable incidence, $M_Q=O(N_\Delta)$ and $d_Q=O(1)$. Without a bound on local-state counts, even a unary factor can produce a dense incompatibility block.

\subsection{QUBO as a special case}

Consider a binary quadratic objective
\begin{equation}
  E(z)=c+\sum_{i=1}^n h_i z_i+
  \sum_{\{i,j\}\in\mathcal E}J_{ij}z_iz_j,
  \qquad z_i\in\{0,1\}.
  \label{eq:qubo}
\end{equation}
Use one unary factor for every retained $h_i z_i$ and one pair factor for every nonzero $J_{ij}z_iz_j$. Constants are handled outside the StQP recovery map.

\begin{corollary}[QUBO-to-StQP]
\label{cor:qubo}
A QUBO with $n_u$ unary factors and $m$ nonzero pair factors admits an exact weighted StQP encoding with
\begin{equation}
  N_\Delta\le2n_u+4m.
  \label{eq:qubosize}
\end{equation}
For a dense $n$-variable QUBO this is $O(n^2)$; for a sparse interaction graph it is $O(n+m)$. The off-diagonal coefficient alphabet remains $\{0,1\}$, while objective values appear in the positive diagonal entries.
\end{corollary}

The compiler expands a pairwise binary objective into local configurations. Consistency among repeated copies of each variable determines the zero/one off-diagonal pattern, while factor costs determine diagonal reciprocals. The resulting dimension can exceed that of the original QUBO substantially.

\subsection{Worked example: weighted Max-Cut}

Let $G=(V,E)$ have nonnegative edge weights $c_e$. A cut is represented by binary labels $z_i\in\{0,1\}$ and has weight
\[
  \operatorname{cut}_c(z)=\sum_{e=\{i,j\}\in E}c_e\,\mathbf1[z_i\ne z_j].
\]
For each edge $e=\{i,j\}$, create four local states $(e,a,b)$ with $(a,b)\in\{0,1\}^2$. Two edge states are compatible if they agree on every shared endpoint. Give state $(e,a,b)$ weight
\begin{equation}
  w_{e,a,b}=L+c_e\mathbf1[a\ne b],
  \qquad L>0.
  \label{eq:maxcutweights}
\end{equation}
A full consistent clique contains one state per edge and has weight
\[
  L|E|+\operatorname{cut}_c(z).
\]
Therefore
\begin{equation}
  \omega_w(H)=L|E|+\operatorname{MaxCut}_c(G).
  \label{eq:maxcutclique}
\end{equation}
With $w_{\min}=L$, \Cref{thm:weightedclique} gives
\begin{equation}
  \nu(Q)=\frac{\tau L}{L|E|+\operatorname{MaxCut}_c(G)}.
  \label{eq:maxcutstqp}
\end{equation}
For an unweighted graph and $L=1$,
\begin{equation}
  \operatorname{MaxCut}(G)=\frac{\tau}{\nu(Q)}-|E|.
  \label{eq:maxcutrecover}
\end{equation}
The formulation uses exactly $4|E|$ simplex coordinates before optional preprocessing.

\subsection{Approximation distortion}
\label{subsec:approximation-distortion}

Let $Q_E$ be the factor-state matrix of \Cref{thm:factorcompiler},
let $\nu=\nu(Q_E)>0$, and retain the recovery constants $B$, $\tau$, and
$w_{\min}$ from \eqref{eq:factorrecover}. For a positive scalar estimate
$\widehat\nu=\nu+\delta$, define the energy proxy
\[
  \widehat E:=B-\frac{\tau w_{\min}}{\widehat\nu}.
\]
The reciprocal recovery formula gives the exact identity
\begin{equation}
  \bigl|\widehat E-E^\star\bigr|
  =\frac{\tau w_{\min}|\delta|}{\nu\,\widehat\nu}.
  \label{eq:reciprocalerror}
\end{equation}
For a scalar estimate alone, $\widehat E$ need not be the energy of any
discrete assignment, and $\delta$ can have either sign. A discrete-solution
guarantee requires both a feasible simplex point and a specified decoder.

Suppose instead that $x\in\Delta_{N_\Delta}$ is returned and that
$\widehat\nu=x^{\mathsf T}Q_E x$ is evaluated using the intended matrix.
Then $\delta\ge0$. Repeatedly eliminate a nonedge from the positive support
by moving all its mass to whichever endpoint gives the smaller objective,
as in Appendix~\ref{app:decoding}. The weighted purification argument in
Theorem~\ref{thm:weightedclique} makes each such step objective-nonincreasing.
Minimize on the resulting clique face, extend the clique to a maximal
clique, and minimize on the enlarged face. Positive vertex weights ensure
that extension cannot increase the face-minimum value.

Under the complete-state hypothesis of \Cref{thm:factorcompiler},
every maximal clique contains one consistent state from every factor and
therefore determines an assignment $z$. Its weighted barycenter $y$ satisfies
\[
  y^{\mathsf T}Q_Ey
  =\frac{\tau w_{\min}}{B-E(z)}
  \le x^{\mathsf T}Q_Ex=\widehat\nu.
\]
Since $t\mapsto B-\tau w_{\min}/t$ is increasing for $t>0$, this decoder gives
\[
  E^\star\le E(z)\le\widehat E,
  \qquad
  0\le E(z)-E^\star
  \le\frac{\tau w_{\min}\delta}{\nu(\nu+\delta)}.
\]
The bound does not automatically apply to a decoder based on voting,
coordinate ranking, or another repair rule that can increase the StQP
objective.

Writing $W=B-E^\star=\tau w_{\min}/\nu$, a multiplicative guarantee
$x^{\mathsf T}Q_Ex\le\rho\nu$, with $\rho\ge1$, similarly yields
\[
  E(z)-E^\star\le
  \left(1-\frac{1}{\rho}\right)W.
\]
Thus the factor-state route preserves a ratio for the positive transformed
objective $B-E$, but generally not the same ratio for $E$ itself, which
may vanish or be negative. Both additive and multiplicative guarantees for
the original energy depend on the offsets through $W$. If a constant $c$
was omitted before compiling, the original objective is $c+E(z)$; restoring
$c$ changes neither the assignment nor its additive suboptimality. These
bounds concern the intended objective evaluated on a feasible simplex
point; coefficient and evaluation errors must be included separately.

\subsection{Objective separation}
\label{subsec:objective-separation}

For integer-valued factor tables, the reciprocal recovery map gives a
separation between distinct assignment values. This is a gap between
local-minimum values, not a positive separation between the global minimum
and every other feasible simplex point.

\begin{proposition}[Factor-model objective separation]
\label{prop:weightedgap}
Assume the hypotheses of \Cref{thm:factorcompiler}, with
integer-valued $\phi_\alpha$ and integer offsets $B_\alpha$ satisfying
\eqref{eq:Balpha}. Let
$W:=\omega_w(H_E)=B-E^\star$ and assume $W\ge2$.
Every nonglobal local minimum of $\mathrm{StQP}(Q_E)$ has value at least
$\nu(Q_E)+g_W$, where
\begin{equation}
  g_W:=\frac{\tau w_{\min}}{W-1}
       -\frac{\tau w_{\min}}{W}
       =\frac{\tau w_{\min}}{W(W-1)}.
  \label{eq:factorgap}
\end{equation}
The corresponding relative separation is at least
$g_W/\nu(Q_E)=1/(W-1)$. Equality is attained when a maximal clique has
weight $W-1$; otherwise the separation between distinct local-minimum
values is larger. If all local minima are global, the statement about
nonglobal minima is vacuous.
\end{proposition}

\begin{proof}
Each vertex weight $w_{\alpha,a}=B_\alpha-\phi_\alpha(a)$ is a positive
integer. The weighted landscape theorem gives a local-minimum value
$\tau w_{\min}/w(C)$ for each maximal clique $C$. A nonglobal clique has
integer weight $w(C)\le W-1$, so its value is at least
$\tau w_{\min}/(W-1)$. Subtracting
$\nu(Q_E)=\tau w_{\min}/W$ proves the claim.
\end{proof}

Under the complete-state hypothesis, $W=1$ is degenerate: positive integer
weights imply that there is only one factor and all its states have weight
one, so every assignment is optimal. There is then no nonglobal
assignment value to separate.

More generally, suppose every nonglobal assignment has energy at least
$E^\star+\Delta_E$, where $0<\Delta_E<W$. The same argument gives
\begin{equation}
  g_{W,\Delta_E}:=
  \frac{\tau w_{\min}\Delta_E}{W(W-\Delta_E)}.
  \label{eq:spacinggap}
\end{equation}
This is exact when the level $E^\star+\Delta_E$ is attained. For integer
factor tables, a positive common divisor of all within-table energy
differences provides a computable spacing bound. Unlike the unit-step estimate,
\eqref{eq:spacinggap} is invariant when all weights and the known energy
spacing are multiplied by a common positive constant, as required by
\Cref{rem:weightscaling}.

\paragraph{Interpretation of separation bounds.}
\label{par:gap-interpretation}
The quantities $g_W$ and $g_{W,\Delta_E}$ are guaranteed lower bounds on the difference between the optimum and a nonglobal local-minimum value. They are exact only when the relevant next assignment level is attained. They do not separate the optimum from every other feasible simplex point. Accuracy tolerances derived from these bounds are sufficient guarantees for the specified decision or decoding procedure, rather than necessary resolution requirements for every instance. A smaller bound alone therefore establishes neither a smaller actual gap nor a change in stationary points or basins. Throughout the offset, penalty, and hardware comparisons, we use ``separation bound'' in this sense.

For a decision threshold $E^\star\le T$, with integer $T$ and
$K:=B-T\ge2$, the possible optimum values satisfy
\[
  E^\star\le T
  \ \Longrightarrow\ \nu(Q_E)\le\frac{\tau w_{\min}}{K},
  \qquad
  E^\star>T
  \ \Longrightarrow\ \nu(Q_E)\ge\frac{\tau w_{\min}}{K-1}.
\]
Thus $g_K=\tau w_{\min}/[K(K-1)]$ is a lower bound on the separation of
these decision classes. Proposition~\ref{prop:perturb} implies that a
midpoint test on the optimal value remains valid when the sum of the
entrywise coefficient-error bound and the two-sided optimal-value
estimation error is less than $g_K/2$. A feasible objective value without
a suboptimality bound does not by itself provide such a test.

The same uniform perturbation bound gives a useful decoding condition.
Suppose $\widetilde Q_E=Q_E+\Delta Q$,
$\|\Delta Q\|_{\max}\le\epsilon$, and a feasible point satisfies
\[
  x^{\mathsf T}\widetilde Q_Ex
  \le\nu(\widetilde Q_E)+r,
  \qquad r\ge0.
\]
Then $x^{\mathsf T}Q_Ex\le\nu(Q_E)+r+2\epsilon$.
Consequently, $r+2\epsilon<g_W$ is sufficient for the
objective-nonincreasing decoder of \Cref{subsec:approximation-distortion} to recover an optimal
assignment. A valid $g_{W,\Delta_E}$ can replace $g_W$ in this condition.

\begin{remark}[Choice of offsets and penalties]
\label{rem:offsets}
Let $F:=|\mathcal F|\ge1$, and choose a common positive integer slack
$\eta$ in the offsets:
\[
  B_\alpha=\max_{a\in\mathcal A_\alpha}\phi_\alpha(a)+\eta,
  \qquad
  \Phi:=\sum_{\alpha\in\mathcal F}
                \max_{a\in\mathcal A_\alpha}\phi_\alpha(a),
  \qquad D:=\Phi-E^\star\ge0.
\]
Then $w_{\min}=\eta$ and $W=D+F\eta$. For $W\ge2$, the unit-step bound in
\eqref{eq:factorgap} is exactly
\begin{equation}
  g(\eta)
  =\frac{\tau\eta}{(D+F\eta)(D+F\eta-1)}.
  \label{eq:offsetgap}
\end{equation}
If $D=0$, every assignment attains the same energy and there is no
nonoptimal assignment level to separate; $\eta=1$ is admissible.
For $D\ge1$, differentiation with respect to a real $\eta>0$ gives
\[
  g'(\eta)=
  \frac{\tau\bigl[D(D-1)-F^2\eta^2\bigr]}
       {(D+F\eta)^2(D+F\eta-1)^2}.
\]
The integer maximizing this unit-step bound can be selected by evaluating
it at the two neighboring integers of the constrained continuous maximizer:
\begin{equation}
  \begin{gathered}
    \eta_c:=\max\left\{1,\frac{\sqrt{D(D-1)}}{F}\right\},
    \qquad
    \mathcal J:=\{\lfloor\eta_c\rfloor,\lceil\eta_c\rceil\},\\[3pt]
    \eta^\star\in\operatorname*{arg\,max}_{j\in\mathcal J}
      \frac{\tau j}{(D+Fj)(D+Fj-1)}.
  \end{gathered}
  \label{eq:offsetrule}
\end{equation}
Nearest-integer rounding of $D/F$ is not an exact selection rule.
For example, with $D=19$ and $F=13$, $g(1)=\tau/992$ but
$g(2)=\tau/990$.

For fixed $D\ge1$ and fixed $\tau$, sufficiently large $F$ gives
$\eta^\star=1$ and $g(1)=\Theta(F^{-2})$.
When $D=\Theta(F)$, the optimized bound likewise scales as
$\Theta(F^{-2})$.
In the regime $D/F\to\infty$, the interior approximation gives
$\eta^\star\sim D/F$ and $g(\eta^\star)\sim\tau/(4DF)$.
For fixed $D,F$ and excessive slack, $g(\eta)\sim\tau/(F^2\eta)$ as
$\eta\to\infty$.

The exact $D$ is not generally available before solving. A certified
lower bound $E_{\rm lb}\le E^\star$ gives the integer upper bound
$\overline D:=\lceil\Phi-E_{\rm lb}\rceil\ge D$.
Since \eqref{eq:offsetgap} decreases as $D$ increases, substituting
$\overline D$ for $D$ yields a conservative gap bound, and
\eqref{eq:offsetrule} can be used to maximize that bound.
The factorwise bound
$E_{\rm lb}=\sum_\alpha\min_a\phi_\alpha(a)$ is always available;
a relaxation may improve it. A feasible incumbent supplies an upper
bound on $E^\star$, not the lower bound required here.

Let
$R_{\max}:=\max_\alpha(\max_a\phi_\alpha(a)-\min_a\phi_\alpha(a))$.
The diagonal coefficient range is
$w_{\max}/w_{\min}=1+R_{\max}/\eta$.
When the matrix contains a nonedge entry equal to one, its full nonzero
coefficient range is
$\kappa_{\rm coef}=(\eta+R_{\max})/(\tau\eta)$;
otherwise this expression is an upper bound.
Increasing $\eta$ therefore improves coefficient range while eventually
reducing the gap bound.

For hard constraints encoded by a common penalty $P$, retain the
sufficient condition from the complete-state discussion:
$P>S_{\rm soft}$, provided a feasible assignment exists. For integer
soft-factor data, $P=S_{\rm soft}+1$ satisfies this condition. Once feasibility is
ensured, increasing $P$ leaves the optimum over feasible assignments
unchanged but can increase $\Phi$, $D$, and the coefficient range.
Its effect on the separation bound follows from \eqref{eq:offsetgap},
with the interpretation given above; \Cref{subsec:bilpprecision} illustrates
the dependence on energy spacings.
\end{remark}

\section{Selection, packing, covering, and matching formulations}
\label{sec:packing}

The graph compiler is particularly direct when a feasible solution is a family of pairwise compatible objects. This section collects the most economical formulations of that kind.

\subsection{Clique, independent set, and vertex cover}
\label{subsec:cliqueis}

Maximum clique is represented directly by \Cref{thm:graphcompiler}. The maximum independent-set number $\alpha(G)$ is the clique number of the complement graph $\overline G$, so
\begin{equation}
  \alpha(G)=\frac{1}{\nu(\Qhat(\overline G))}.
  \label{eq:independent}
\end{equation}
A vertex cover $C$ has complement $V\setminus C$ independent. Thus $G$ has a vertex cover of size at most $k$ if and only if
\begin{equation}
  \nu(\Qhat(\overline G))\le\frac{1}{n-k}.
  \label{eq:vertexcover}
\end{equation}
These formulations use $n$ coordinates and no auxiliary variables.

\subsection{Set packing}

Given a family $\mathcal S=\{S_1,\ldots,S_m\}$, create one graph vertex per set and join $i$ and $j$ when $S_i\cap S_j=\varnothing$. Then cliques are pairwise disjoint subfamilies.

\begin{proposition}[Set packing]
The maximum packing cardinality is
\begin{equation}
  p^\star=\frac{1}{\nu(\Qhat(H_{\mathcal S}))}.
  \label{eq:setpacking}
\end{equation}
In particular, a packing of at least $k$ sets exists if and only if the StQP optimum is at most $1/k$. The formulation uses $m$ simplex coordinates.
\end{proposition}

A weighted set-packing problem is obtained by assigning a positive weight to each candidate set and applying \Cref{thm:weightedclique}.

\subsection{Three-dimensional matching}

Let $T\subseteq X\times Y\times Z$ be a collection of triples. Create one graph vertex per triple and connect two vertices when the triples are disjoint in all three coordinates. A clique is a three-dimensional matching.

\begin{proposition}[3D matching]
The maximum matching cardinality equals $1/\nu(\Qhat(H_T))$. If $|X|=|Y|=|Z|=q$, a perfect three-dimensional matching exists if and only if
\[
  \nu(\Qhat(H_T))=\frac1q.
\]
The formulation uses $|T|$ coordinates.
\end{proposition}

\subsection{Exact cover through weighted clique}

Let $U$ be a universe and $\mathcal S=\{S_1,\ldots,S_m\}$ a family of nonempty subsets. Form the disjointness graph used for set packing, but give vertex $i$ weight $w_i=|S_i|$. A clique is a disjoint subfamily, and its total weight equals the cardinality of its union. Therefore every clique has weight at most $|U|$.

\begin{proposition}[Exact cover]
\label{prop:exactcover}
The family $\mathcal S$ contains an exact cover of $U$ if and only if
\begin{equation}
  \omega_w(H_{\mathcal S})=|U|.
  \label{eq:exactcoverclique}
\end{equation}
Equivalently, with $w_{\min}=\min_i|S_i|$,
\begin{equation}
  \nu(Q_{\tau,w}(H_{\mathcal S}))
  =\frac{\tau w_{\min}}{|U|}.
  \label{eq:exactcoverstqp}
\end{equation}
A clique-supported minimizer decodes to the exact-cover subfamily.
\end{proposition}

The natural binary comparison is maximum-weight independent set on the conflict graph of intersecting sets. Let
\[
  \mathcal E_c:=\{\{i,j\}:i<j,\ S_i\cap S_j\ne\varnothing\},
  \qquad r_i:=\frac{w_i}{w_{\max}}\in(0,1].
\]
The standard independent-set penalty construction \citep[Section~4.2]{Lucas2014} gives
\begin{equation}
  H_{\rm EC}(z)=-\sum_{i=1}^{m}r_i z_i
       +2\sum_{\{i,j\}\in\mathcal E_c}z_i z_j,
  \qquad z\in\{0,1\}^{m}.
  \label{eq:exactcover-qubo}
\end{equation}
If a selected set $i$ has $d_i\ge1$ selected conflicts, deleting it changes $H_{\rm EC}$ by $r_i-2d_i<0$. Every minimizer is therefore conflict-free. On a disjoint subfamily the objective equals minus its covered cardinality divided by $w_{\max}$, so an exact cover exists precisely when
\[
  \min_z H_{\rm EC}(z)=-\frac{|U|}{w_{\max}}.
\]
Both formulations use $m$ variables or coordinates and one pair interaction per intersecting set pair. The QUBO has a fixed conflict coefficient after normalizing its unary weights; the StQP stores reciprocal weights on its diagonal and uses one global normalization. Thus neither a dimension saving nor a general penalty-scale advantage follows from this comparison. In the displayed QUBO polynomial convention, if a conflict exists, the nonzero coefficient ratio is $2w_{\max}/w_{\min}$, equal to the full StQP coefficient ratio at $\tau=1/2$. Converting the polynomial to a symmetric $z^{\mathsf T}Qz$ convention halves its off-diagonal coefficients, so coefficient comparisons must specify the convention.

\subsection{Hitting set and set cover}

Karp's \textsc{hitting set} asks for a set $W\subseteq S$ that meets every member of a family $U_1,\ldots,U_m\subseteq S$ in \emph{exactly} one element \citep{Karp1972}. It is exact cover on the transposed incidence structure, and therefore weighted-graph-native. Assume every $U_i$ is nonempty, and discard elements that lie in no $U_i$, since they never affect feasibility. Give each remaining element $s$ the weight $\deg(s):=|\set{i:s\in U_i}|$, and join $s\ne s'$ when no $U_i$ contains both.

\begin{proposition}[Karp's hitting set]
\label{prop:karphitting}
Let $H_U$ be this weighted graph. A set $W$ with $|W\cap U_i|=1$ for all $i$ exists if and only if $\omega_w(H_U)=m$, equivalently $\nu(Q_{\tau,w}(H_U))=\tau w_{\min}/m$ with $w_{\min}=\min_s\deg(s)$.
\end{proposition}

\begin{proof}
A clique $W$ meets every $U_i$ at most once, so $w(W)=\sum_{s\in W}\deg(s)=\sum_i|W\cap U_i|\le m$, with equality exactly when every $U_i$ is met once. The value statement is \Cref{thm:weightedclique}.
\end{proof}

The Garey--Johnson variant \citep{GareyJohnson1979}, which asks for at most $k$ elements meeting every $U_i$ \emph{at least} once, is a covering problem and is treated next.

The Garey--Johnson hitting set and set cover do not generally reduce to pairwise compatibility as directly as packing. A feasible hitting set is not a family of pairwise compatible elements; its defining condition is collective coverage. Two routes are available:

\begin{enumerate}[label=(\roman*)]
\item reduce set cover to hitting set or vice versa and then to a binary linear system;
\item build a local-state factor graph in which each coverage constraint is a factor over the incident selection variables, encoding infeasible coverage patterns as high-cost states rather than omitting them (\Cref{rem:nopruning}).
\end{enumerate}

The second route can be compact when each set or element has bounded incidence, but a factor of degree $d$ has $2^d$ local states in the factor-state compiler, none of which may be deleted (\Cref{rem:nopruning}). The binary-linear compiler of \Cref{sec:bilp} remains polynomial in the ordinary bit model and is preferable for unbounded incidence.

\subsection{Dominant sets and weighted graph structure}

Continuous simplex quadratics are also used natively for weighted clustering. Maximizing $x^{\mathsf T}Ax$ over the simplex can identify coherent weighted vertex groups, with local solutions related to dominant sets \citep{PavanPelillo2007}, with edge-weighted and hypergraph extensions \citep{pavan-pelillo-edgeweighted-2003,rotabulo-pelillo-hypergraph-2009}. Here the simplex quadratic is the model itself rather than a reduction from a discrete NP problem, and it belongs in this survey because the same StQP hardware primitive can support both exact combinatorial encodings and native continuous graph models.

\section{A generic binary-linear feasibility compiler}
\label{sec:bilp}

Direct compatibility formulations are preferable when available. For completeness, we now give a generic polynomial route from binary linear feasibility to StQP. It covers every problem with a polynomial-size 0--1 integer programming formulation and supplies a uniform route for the remaining members of Karp's list.

\subsection{From inequalities to a zero-penalty QUBO}

Consider
\begin{equation}
  Az\le b,
  \qquad
  z\in\{0,1\}^n,
  \label{eq:bilp}
\end{equation}
with integer data. Rational data can be scaled to integers with polynomial bit-length growth. For row $r$, let
\begin{equation}
  L_r:=\sum_{i:a_{ri}<0}a_{ri},
  \qquad
  U_r:=b_r-L_r.
  \label{eq:slackrange}
\end{equation}
If $U_r<0$, the row is infeasible for every binary $z$ and the instance can be rejected immediately. Otherwise, if the row is feasible, its slack $s_r=b_r-a_r^{\mathsf T}z$ is an integer in $[0,U_r]$. Introduce
\begin{equation}
  \ell_r:=\left\lceil\log_2(U_r+1)\right\rceil
  \label{eq:slackbits}
\end{equation}
slack bits $y_{r0},\ldots,y_{r,\ell_r-1}$ and encode
\[
  s_r(y)=\sum_{t=0}^{\ell_r-1}2^t y_{rt}.
\]
Define
\begin{equation}
  P(z,y):=\sum_{r=1}^m
  \left(a_r^{\mathsf T}z+s_r(y)-b_r\right)^2.
  \label{eq:bilppenalty}
\end{equation}
Because binary variables satisfy $u^2=u$, expansion of \eqref{eq:bilppenalty} is a QUBO \citep{BorosHammer2002}; higher-degree residuals can be quadratized with few ancillas \citep{boros-crama-rodriguezheck-compact-quadratizations-2020}, and penalty weights admit alternatives to the dominant-coefficient rule \citep{montanez-barrera-unbalanced-penalization-2024}.

\begin{theorem}[Binary-linear feasibility compiler]
\label{thm:bilpcompiler}
The system \eqref{eq:bilp} is feasible if and only if
\begin{equation}
  \min_{z,y\in\{0,1\}}P(z,y)=0.
  \label{eq:bilpzero}
\end{equation}
Combining the expanded QUBO with \Cref{thm:factorcompiler} gives a polynomial-size StQP threshold encoding of binary-linear feasibility.
\end{theorem}

\begin{proof}
If $z$ is feasible, every integer slack $b_r-a_r^{\mathsf T}z\in[0,U_r]$ has a binary representation using $\ell_r$ bits, and all squared residuals vanish. Conversely, if $P=0$, every residual vanishes and $a_r^{\mathsf T}z=b_r-s_r(y)\le b_r$.
\end{proof}

Equalities require no slack. A constraint $a^{\mathsf T}z\ge b$ is multiplied by $-1$. Bounded integer variables are encoded in binary together with range constraints if the binary box is larger than the intended domain.

\subsection{Size accounting}

Let
\begin{equation}
  p=n+\sum_{r=1}^m\ell_r
  \label{eq:totalbits}
\end{equation}
be the number of binary variables after adding slack bits. Let $t_r$ be the number of binary variables appearing in residual $r$. Expansion of $P$ contains at most
\begin{equation}
  M_2\le\sum_{r=1}^m\binom{t_r}{2}
  \label{eq:pairterms}
\end{equation}
nonzero pair terms before duplicate aggregation. The basic factor-state compiler therefore uses at most
\begin{equation}
  N_\Delta\le2p+4M_2.
  \label{eq:bilpsize}
\end{equation}
This is polynomial in the binary input length because each $\ell_r$ is logarithmic in an integer whose bit length is part of the input. The bound is intentionally simple rather than tight. Factor aggregation, quadratization reuse, and direct CSP constructions can reduce it substantially.

\subsection{Adding an objective or decision threshold}

A 0--1 optimization problem
\[
  \min\set{c^{\mathsf T}z:Az\le b,\ z\in\{0,1\}^n}
\]
can be converted to a sequence of feasibility questions $c^{\mathsf T}z\le T$, or to one QUBO using a sufficiently dominant feasibility penalty. In the second route, minimize $MP(z,y)+c^{\mathsf T}z$ with
\[
  M>\sum_{i=1}^n|c_i|,
\]
which suffices because $P\ge1$ at every infeasible point while the objective varies by at most $\sum_i|c_i|$. For integer objective data, the feasibility-sequence route avoids that coefficient and preserves exact decision thresholds, using $O(\log(2+\sum_i|c_i|))$ StQP solves after checking feasibility.

\subsection{Precision of the binary-linear route}
\label{subsec:bilpprecision}

The residual squares in \eqref{eq:bilppenalty} carry the input numbers into the factor values. The resulting matrix still has polynomial encoding length, but its coefficient ratios and value separations depend on the factorization, offsets, and attainable energy levels. We use the separation bounds and accuracy conditions of \Cref{subsec:objective-separation}.

For Partition with $n\ge2$ positive items, let $A=\sum_i a_i$. Expanding $(2\sum_i a_i z_i-A)^2$ gives the constant $A^2$, unary coefficients $4a_i(a_i-A)$, and pair coefficients $8a_i a_j$. Retaining these $F=n(n+1)/2$ factors gives, on a yes-instance,
\[
  D=5A^2-4\sum_i a_i^2.
\]
When $D/F\to\infty$, the offset rule \eqref{eq:offsetrule} yields the unit-step estimate $g(\eta^\star)\sim\tau/(4DF)$, whereas unit slack gives exactly $g(1)=\tau/[(D+F)(D+F-1)]$. These bounds can be inverse-polynomial in integer magnitudes, and hence exponentially small in their bit lengths. A tolerance of $2^{-O(b)}$ requires $O(b)$ binary precision bits; physical resolution and bit complexity are distinct costs.

A common scaling illustrates the distinction. Take $a_1=a_2=L$ and slack $\eta=4L^2$, where $L\ge1$ is an integer. Each unary block has weights $(4,8)L^2$, and the pair block has weights $(12,12,12,4)L^2$. Thus $w_{\min}=4L^2$ and $W=24L^2$. The two distinct local-minimum values are $\tau/6$ and $\tau/5$, so the actual gap is $\tau/30$, independent of $L$. The matrix is also independent of $L$, with full nonzero coefficient range $3/\tau$. Nevertheless, the unit-step bound is
\[
  g_W=\frac{\tau}{6(24L^2-1)}\longrightarrow0.
\]
Using the attainable spacing $\Delta_E=4L^2$ in \eqref{eq:spacinggap} recovers the exact gap. Integer magnitudes alone therefore do not determine the physical precision burden.

For a dominant feasibility penalty $M$, if $F$ and $\tau$ are fixed and $D(M)=\Theta(M)$, optimizing the common slack makes the unit-step bound scale as $\Theta(M^{-1})$; fixed slack instead gives $\Theta(M^{-2})$. These scalings apply to the bound defined in \eqref{eq:offsetgap}.

An alternative is to encode arithmetic in Boolean constraints and compose a polynomial reduction to 3-SAT with \Cref{con:satgraph}. For $m\ge2$ clauses this gives the fixed alphabet $\{0,1,2\}$ and a target decision-class separation of at least $1/[m(m-1)]$, with $m$ polynomial in input bit length. This offers an inverse-polynomial sufficient tolerance, but there is no universal dimension ordering between this route and the residual-square route. The comparison must report matrix size, coefficient ratios, and a valid separation bound for the chosen representation.

\subsection{Representative constructions}
\label{subsec:bilp-constructions}

\paragraph{Subset sum and Karp's knapsack variant.}
Given positive integers $a_1,\ldots,a_n$ and target $t$, define
\begin{equation}
  P(z)=\left(\sum_{i=1}^n a_i z_i-t\right)^2.
  \label{eq:subsetsum}
\end{equation}
Then $P^\star=0$ exactly when a target-sum subset exists. Expanding $P$ gives a QUBO, dense in the worst case, to which \Cref{cor:qubo} applies.

\paragraph{Partition.}
Let $A=\sum_i a_i$. A partition into equal sums exists exactly when
\begin{equation}
  \min_{z\in\{0,1\}^n}
  \left(2\sum_i a_i z_i-A\right)^2=0.
  \label{eq:partition}
\end{equation}

\paragraph{Set cover.}
Use one bit $z_j$ per candidate set. The decision whether at most $k$ sets cover $U$ is
\begin{equation}
  \sum_{j:u\in S_j}z_j\ge1\quad(u\in U),
  \qquad
  \sum_jz_j\le k.
  \label{eq:setcoverbilp}
\end{equation}

\paragraph{Hitting set (Garey--Johnson variant).}
Use one bit $z_u$ per universe element and impose
\begin{equation}
  \sum_{u\in S_j}z_u\ge1\quad(j=1,\ldots,m),
  \qquad
  \sum_{u\in U}z_u\le k.
  \label{eq:hittingbilp}
\end{equation}

\paragraph{Job sequencing (Karp).}
Karp's \textsc{job sequencing} has one machine and jobs $j\in[n]$ with integer execution times $T_j$, deadlines $D_j$ and penalties $P_j$; it asks for an order in which the total penalty of late jobs is at most $k$ \citep{Karp1972}. Let $y_j=1$ mean that job $j$ is on time. A set of on-time jobs can be scheduled before their deadlines exactly when earliest-deadline-first order does so, which gives
\begin{equation}
  \sum_{i:D_i\le D_j}T_i\,y_i\le D_j\quad(j\in[n]),
  \qquad
  \sum_{j=1}^nP_j(1-y_j)\le k.
  \label{eq:jobkarp}
\end{equation}
The late jobs are appended in any order. Using \eqref{eq:slackbits}, this uses $p=n+\sum_j\lceil\log_2(D_j+1)\rceil+\lceil\log_2(k+1)\rceil$ bits.

\paragraph{Parallel-machine job sequencing (Lucas variant).}
For the integer-length makespan version used in Ising formulation work \citep{Lucas2014}, let job $i$ have length $L_i$, let $a\in[M]$ index machines, and let $T$ be the makespan threshold. Binary assignment variables $z_{ia}$ satisfy
\begin{equation}
  \sum_{a=1}^{M} z_{ia}=1\quad(i\in[n]),
  \qquad
  \sum_{i=1}^nL_i z_{ia}\le T\quad(a\in[M]).
  \label{eq:jobbilp}
\end{equation}
This uses $nM$ assignment bits plus $\lceil\log_2(T+1)\rceil$ slack bits for each machine-load constraint.

\paragraph{Feedback arc set.}
For a directed graph, let $y_{uv}=1$ mean arc $(u,v)$ is deleted. Assign each vertex an integer order $\pi_v\in\{0,\ldots,n-1\}$, binary encoded. Impose
\begin{equation}
  \pi_u+1\le \pi_v+n y_{uv}\quad((u,v)\in E),
  \qquad
  \sum_{(u,v)\in E}y_{uv}\le k.
  \label{eq:feedbackarc}
\end{equation}
If at most $k$ arcs can be removed to make the graph acyclic, the remaining arcs admit such a topological order; the converse is immediate. Each arc row has slack range $O(n)$ and needs $O(\log n)$ slack bits, so $p=\Theta((n+|E|)\log n)$.

\paragraph{Feedback vertex set.}
Let $y_v=1$ mean vertex $v$ is removed. One valid order formulation is
\begin{equation}
  \pi_u+1\le \pi_v+n(y_u+y_v)
  \quad((u,v)\in E),
  \qquad
  \sum_vy_v\le k.
  \label{eq:feedbackvertex}
\end{equation}
Every remaining arc then respects a topological order. As for arc deletion, $p=\Theta((n+|E|)\log n)$.

\paragraph{Steiner tree.}
Let $R\subseteq V$ be terminals, let edge costs satisfy $c_e\ge0$, and choose a root $r\in R$. Use binary edge variables $y_e$ and directed flow variables $f_{uv}$, binary encoded, with
\begin{align}
  0\le f_{uv}&\le(|R|-1)y_{\{u,v\}},\label{eq:steinerlink}\\
  \sum_{v}f_{rv}-\sum_vf_{vr}&=|R|-1,\label{eq:steinerroot}\\
  \sum_vf_{vu}-\sum_vf_{uv}&=1\quad(u\in R\setminus\{r\}),\label{eq:steinerterm}\\
  \sum_vf_{vu}-\sum_vf_{uv}&=0\quad(u\notin R),\label{eq:steinernonterm}\\
  \sum_ec_ey_e&\le K.\label{eq:steinercost}
\end{align}
A feasible solution contains a connected subgraph spanning all terminals; cycles can be removed without increasing nonnegative cost. With $\lceil\log_2|R|\rceil$ bits per directed flow, $p=|E|+O(|E|\log|R|)+\lceil\log_2(K+1)\rceil$.

These formulations are standard integer-programming patterns rather than claims of novel polyhedral strength. Their role here is to close the compilation chain from Karp's problems to StQP.

\section{A catalog for Karp's 21 NP-complete problems}
\label{sec:karp}

Karp's 1972 paper established the polynomial equivalence of a broad class of satisfiability, covering, packing, routing, assignment, sequencing, and numerical problems \citep{Karp1972}. Lucas later provided explicit Ising constructions for all 21 problems and analyzed spin counts and coupling requirements \citep{Lucas2014}. \Cref{tab:karp21} summarizes the corresponding StQP routes developed in this paper. Rather than reproduce Karp's original reduction tree, we choose modern formulations when they are more direct. Linear-growth interreductions among subsets of the 21 problems may provide alternative compilers with lower overhead for particular applications \citep{FilarHaythorpeTaylor2018}; standard statements of the problems follow \citet{GareyJohnson1979}. The Hamiltonian rows use the position--vertex compatibility construction of \Cref{sec:csp}. Controlled-Markov-chain formulations provide a different continuous approach to Hamiltonicity \citep{borkar-ejov-filar-digraphs-2004}.

Among the 21 original numbered rows, ten use unweighted compatibility graphs and two use weighted graphs. Six use binary-linear feasibility followed by residual-square QUBO and factor-state compilation; the knapsack target-sum and Partition rows start directly from residual-square QUBOs, and Max-Cut uses a factor-state formulation directly. Rows $15'$ and $19'$ give additional variants and are not included in these counts. The catalog therefore establishes coverage through several routes, with different dimension and precision costs.

For the table, ``direct graph'' refers to \eqref{eq:integergraph}, ``weighted graph'' to \eqref{eq:weightedQ}, and ``BILP'' to \Cref{thm:bilpcompiler} followed by \Cref{thm:factorcompiler}. For every residual-square route write $P=c+E$, where $c$ is the omitted constant and $E$ is the compiled factor energy. The original penalty is recovered as
\begin{equation}
  P^\star=c+B-\frac{\tau w_{\min}}{\nu(Q_E)};
  \qquad P^\star=0\ \Longleftrightarrow\ E^\star=-c.
  \label{eq:catalog-constant}
\end{equation}
Trivial instances and out-of-range thresholds are handled before compilation.

The size column (``Basic $N_\Delta$'') reports the number of simplex coordinates before problem-specific simplification. For generic binary-linear routes, $p$ is the number of binary variables after slack and bounded-integer encoding, and $M_2$ is the number of nonzero pair terms in the resulting zero-penalty QUBO; the bound $2p+4M_2$ follows from \eqref{eq:bilpsize}. The column gives $p$ for rows 7, 8, 16, 19 and 19$'$, where the encoding fixes it; for rows 2, 6 and 15$'$ it depends on the instance's coefficient ranges through \eqref{eq:slackrange}. A row over $q$ variables contributes up to $\binom q2$ pair terms, so a single cardinality budget $\sum y\le k$ already makes $M_2=\Theta(q^2)$. These generic bounds are polynomial but are not necessarily optimal.

\begin{landscape}
\begingroup
\setlength{\tabcolsep}{3pt}
\renewcommand{\arraystretch}{1.18}
\setlength{\LTleft}{0pt}
\setlength{\LTright}{0pt plus 1fill}
\setlength{\LTcapwidth}{\linewidth}
\captionsetup{width=\linewidth}
\fontsize{9}{11}\selectfont
\begin{longtable}{@{}
  >{\raggedright\arraybackslash}p{0.035\linewidth}
  >{\raggedright\arraybackslash}p{0.133\linewidth}
  >{\raggedright\arraybackslash}p{0.194\linewidth}
  >{\raggedright\arraybackslash}p{0.152\linewidth}
  >{\raggedright\arraybackslash}p{0.272\linewidth}
  >{\raggedright\arraybackslash}p{0.164\linewidth}
@{}}
\caption{StQP compilation routes for Karp's 21 problems. Rows 15 and 19 use Karp's original definitions; primed rows give additional variants. Route labels, size conventions, and constant recovery are defined in the preceding text.}
\label{tab:karp21}\\
\toprule
\# & Problem & Preferred StQP route & Basic $N_\Delta$
& Exact decision or recovery rule & Main representation issue \\
\midrule
\endfirsthead
\multicolumn{6}{@{}l}{\textbf{Table \thetable\ -- continued}}\\[3pt]
\toprule
\# & Problem & Preferred StQP route & Basic $N_\Delta$
& Exact decision or recovery rule & Main representation issue \\
\midrule
\endhead
\midrule
\multicolumn{6}{r@{}}{Continued on next page}\\
\endfoot
\bottomrule
\endlastfoot
1 & Satisfiability & Clause--literal compatibility graph, Construction~\ref{con:satgraph} & $N_{\rm lit}$, the total number of literal occurrences & SAT iff $\nu(\widehat Q)=1/m$, where $m\ge1$ is the number of clauses. & At the target value, the support selects one consistent literal per clause; local supports may omit clauses. \\[4pt]

2 & 0--1 integer programming & BILP residual squares, \eqref{eq:bilppenalty}, then factor-state StQP & $\le 2p+4M_2$ & Feasible iff $P^\star=0$, equivalently $E^\star=-c$; restore the omitted constant using \eqref{eq:catalog-constant}. & Slack bit length and pair terms generated by squaring residuals. \\[4pt]

3 & Clique & Direct graph compiler on $G$ & $n=|V|$ & $\omega(G)=1/\nu(\widehat Q(G))$. A clique of size at least $k$ exists iff $\nu\le1/k$. & Every maximal clique is a strict local minimum. \\[4pt]

4 & Set packing & Disjointness graph on candidate sets & $m$ candidate sets & Maximum packing cardinality is $1/\nu(\widehat Q)$. & Pairwise intersection tests; positive packing weights require the weighted formulation. \\[4pt]

5 & Vertex cover & Independent set in $G$, equivalently clique in $\overline G$ & $n$ & A cover of size at most $k<n$ exists iff $\nu(\widehat Q(\overline G))\le1/(n-k)$. & Complementation changes the interaction pattern; the decoder returns the complement of an independent set. \\[4pt]

6 & Set covering & BILP constraints \eqref{eq:setcoverbilp} & $\le 2p+4M_2$ & A cover using at most $k$ candidate sets exists iff $P^\star=0$. & High-incidence coverage constraints and the budget row generate pair interactions. \\[4pt]

7 & Feedback node set & Vertex-deletion bits and a binary-encoded topological order, \eqref{eq:feedbackvertex} & $\le 2p+4M_2$; $p=O((n+|E|)\log(n+1))$ & At most $k$ vertices can be deleted to leave a directed acyclic graph iff $P^\star=0$. & Order and slack bits; order coefficients scale with $n$. The deletion budget can generate dense pairs. \\[4pt]

8 & Feedback arc set & Arc-deletion bits and a binary-encoded topological order, \eqref{eq:feedbackarc} & $\le 2p+4M_2$; $p=O((n+|E|)\log(n+1))$ & At most $k$ arcs can be deleted to leave a directed acyclic graph iff $P^\star=0$. & Order encoding and slack bits; the arc budget can couple every pair of deletion variables. \\[4pt]

9 & Directed Hamiltonian circuit & Position--vertex compatibility graph & $n^2$ & A Hamiltonian circuit exists iff $\nu(\widehat Q)=1/n$. & Conflicts enforce distinct positions, distinct vertices, and adjacency in the required direction. \\[4pt]

10 & Undirected Hamiltonian circuit & Position--vertex compatibility graph & $n^2$ & A Hamiltonian circuit exists iff $\nu(\widehat Q)=1/n$. & The same position lift, with undirected adjacency constraints. \\[4pt]

11 & 3-SAT & Clause--literal compatibility graph & $N_{\rm lit}\le3m$ & Satisfiable iff $\nu(\widehat Q)=1/m$. & For $m\ge2$, the adjacent target levels are separated by $1/[m(m-1)]$. \\[4pt]

12 & Chromatic number & Variable--color graph for the decision $\chi(G)\le k$ & $nk$; nontrivial $k<n$ & $k$-colorable iff $\nu(\widehat Q)=1/n$. & Color-label symmetry; finding $\chi(G)$ requires testing candidate values of $k$. \\[4pt]

13 & Clique cover & $k$-colorability of $\overline G$ & $nk$; nontrivial $k<n$ & A cover by at most $k$ cliques exists iff $\nu(\widehat Q)=1/n$. & Complement density and permutation symmetry among cover labels. \\[4pt]

14 & Exact cover & Weighted disjointness graph with $w_j=|S_j|$ & $m$ nonempty candidate sets & Exact cover iff $\omega_w=|U|$, equivalently $\nu(Q_{\tau,w})=\tau w_{\min}/|U|$. & Diagonal coefficients contain reciprocals of set sizes. \\[4pt]

15 & Hitting set: $|W\cap U_i|=1$ (Karp) & Weighted element graph; $w_s=|\{i:s\in U_i\}|$ & $\le |S|$ retained elements & Such a set $W$ exists iff $\omega_w=m$, equivalently $\nu(Q_{\tau,w})=\tau w_{\min}/m$. & Elements sharing a set conflict. Discard elements of zero incidence; weights are incidence degrees. \\[4pt]

$15'$ & Hitting set: at most $k$ elements (Garey--Johnson) & Coverage/budget BILP, \eqref{eq:hittingbilp} & $\le 2p+4M_2$ & A set of at most $k$ elements meeting every member at least once exists iff $P^\star=0$. & Large incidence rows. A full-table factor lift can be large when incidence is unbounded. \\[4pt]

16 & Steiner tree & Flow BILP, \eqref{eq:steinerlink}--\eqref{eq:steinercost} & $\le 2p+4M_2$; $p=O(|E|(1+\log|R|)+\log(K+1))$ & For nonnegative edge costs, a tree spanning $R$ with cost at most $K$ exists iff $P^\star=0$. & Flow and capacity-slack bits, conservation equations, and a cost-budget row. \\[4pt]

17 & 3-dimensional matching & Disjointness graph on triples & $|\mathcal T|$ & For $|X|=|Y|=|Z|=q$, a perfect matching exists iff $\nu(\widehat Q)=1/q$. & Overlapping triples generate conflicts; the support identifies the selected triples. \\[4pt]

18 & Knapsack: target-sum form (Karp) & Square residual \eqref{eq:subsetsum}, then QUBO compiler & $\le2n^2$ & Target sum $t$ is achievable iff $P^\star=0$. Here $c=t^2$, so $E^\star=-t^2$. & Coefficient magnitudes depend on item values and target; pair interactions can be dense. \\[4pt]

19 & Job sequencing: one machine, deadlines and penalties & Earliest-deadline BILP in \Cref{subsec:bilp-constructions} & $\le2p+4M_2$; $\begin{gathered}p=n+\sum_j\ell_j+\ell_k,\\ \ell_j=\lceil\log_2(D_j+1)\rceil,\\ \ell_k=\lceil\log_2(k+1)\rceil\end{gathered}$ & An order with total late-job penalty at most $k$ exists iff $P^\star=0$. & Processing times, deadlines, and late penalties affect bit counts and coefficients. \\[4pt]

$19'$ & Multiprocessor makespan (Lucas) & Assignment/load BILP in \Cref{subsec:bilp-constructions} & $\le2p+4M_2$; $\begin{gathered}p=nM+M\ell_T,\\ \ell_T=\lceil\log_2(T+1)\rceil\end{gathered}$ & A schedule on $M$ machines with makespan at most $T$ exists iff $P^\star=0$. & Machine-label symmetry, numerical job lengths, and load-row pair expansion. \\[4pt]

20 & Partition & Square residual \eqref{eq:partition}, then QUBO compiler & $\le2n^2$ & With $A=\sum_i a_i$, equal partition exists iff $P^\star=0$. Here $c=A^2$, so $E^\star=-A^2$. & Offset choices, weight ratios, and attainable energy spacings determine the numerical burden. \\[4pt]

21 & Max cut & Four local states per edge, \eqref{eq:maxcutweights} & $4|E|$ & $\operatorname{MaxCut}_c(G)=\tau L/\nu(Q)-L|E|$, with $L>0$. & Reciprocal recovery, offset choice, and repeated local copies of vertex labels. \\[4pt]

\end{longtable}
\endgroup
\end{landscape}

\subsection{Direct versus generic routes}

The table shows four distinct regimes.

\paragraph{Graph-native problems.}
Clique, independent set, vertex cover, set packing, three-dimensional matching, coloring, clique cover, Hamiltonian cycle, and satisfiability admit direct unweighted compatibility graphs. Their matrices use coefficients in $\{0,1,2\}$, and an optimal support encodes a maximum compatible selection. For feasibility problems, a satisfying assignment is obtained when that support reaches the specified target size.

\paragraph{Weighted graph-native problems.}
Exact cover, Karp's hitting set, and weighted packing are naturally maximum-weight clique problems. Their off-diagonal entries remain binary, while weights enter through reciprocal diagonal terms. Max-Cut reaches a weighted clique after lifting to local edge configurations; this is the factor-state construction, so \Cref{cor:allassignments} applies to row~21.

\paragraph{Globally constrained problems.}
Set cover, the Garey--Johnson hitting set, feedback sets, Steiner tree, and general 0--1 feasibility involve collective constraints not captured by pairwise compatibility among the original objects. The binary-linear route is polynomial in input bit length. A local-state route is polynomial in explicit table size but can expand succinct high-arity constraints exponentially. Either representation may be too large or numerically demanding for a particular device.

\paragraph{Numerical problems.}
Knapsack, Partition, and the job-sequencing variants carry numerical data into the residual-square coefficients and factor weights. Their precision costs depend on attainable energy spacings, common scaling, and offset choices, as illustrated in \Cref{subsec:bilpprecision}.

\paragraph{Exactness of the routes.}
The direct graph routes recover maximum compatible selections and, when the target threshold is attained, solutions of the original feasibility problem. Their local minima correspond exactly to maximal cliques of the compatibility graph; this need not preserve a different local-move neighborhood on the original variables. The factor-state routes are decision-, value-, and solution-exact for the compiled objective: \Cref{thm:factorcompiler} provides both reciprocal value recovery and an optimal-assignment decoder. On a residual-square feasibility route, restore the omitted constant and require $P^\star=0$ to obtain a feasible witness. These routes do not generally preserve the original discrete local-search landscape, because every assignment becomes a strict local minimum. Approximation guarantees also require a specified objective normalization and decoder; see \Cref{subsec:approximation-theory}.

\subsection{Threshold normalization}

For direct graph formulations, a decision problem typically takes the form $\omega(H)\ge K$. With $Q_\tau$,
\begin{equation}
  \omega(H)\ge K
  \quad\Longleftrightarrow\quad
  \nu(Q_\tau(H))\le\frac{\tau}{K}.
  \label{eq:genericthreshold}
\end{equation}
If the construction guarantees $\omega(H)\le K$, as in SAT, coloring, or Hamiltonian cycle, feasibility is equivalent to the equality $\nu=\tau/K$. Integer scaling at $\tau=1/2$ gives threshold $1/K$.

For weighted formulations, a target $\omega_w(H)\ge W$ is equivalent to
\begin{equation}
  \nu(Q_{\tau,w}(H))\le\frac{\tau w_{\min}}{W}.
  \label{eq:weightedthreshold}
\end{equation}
These formulas provide a direct recovery rule and illustrate the reciprocal compression of large combinatorial values.

\section{Representation cost and comparison with QUBO/Ising}
\label{sec:cost}

Both StQP and QUBO are universal under polynomial reductions, so the relevant comparison is how a particular representation uses the resources available to a solver or physical platform.

\subsection{State semantics}

A QUBO state is a vertex of a binary hypercube. An StQP state is a point of a simplex. In the direct graph compiler, the discrete solution is encoded by the support of $x$, while the normalized magnitudes are determined analytically. In the factor-state compiler, a selected local state receives mass proportional to its positive vertex weight. These representations produce different failure modes:

\begin{itemize}
\item a QUBO device can violate a logical constraint when penalty weights are insufficient;
\item a graph-coded StQP state is always normalized, but an approximate solution may assign small mass to incompatible coordinates;
\item QUBO readout is discrete but may require chain repair after embedding \citep{Choi2011};
\item StQP readout is continuous and requires feasibility and support tolerances.
\end{itemize}

\subsection{Coefficient structure}

The graph compiler has an unusually small coefficient alphabet:
\[
  \widehat q_{ij}\in\{0,1,2\}.
\]
There is no instance-dependent penalty coefficient. The magnitude two always marks incompatibility, zero marks compatibility, and one is the diagonal concentration term. This simplicity should not be confused with numerical robustness: the decision-value separation still decreases with the target support size.

The weighted compiler retains off-diagonal entries in $\{0,1\}$ but uses diagonal values $\tau w_{\min}/w_i$. Thus objective data are localized to onsite coefficients. The factor-state construction can therefore be viewed as a separation between:
\begin{equation}
  \begin{aligned}
  \text{factor-state consistency}&\longleftrightarrow\text{off-diagonal pattern},\\
  \text{objective and constraint penalties}&\longleftrightarrow\text{diagonal values}.
  \end{aligned}
  \label{eq:logicobjective}
\end{equation}

QUBO generally distributes both constraint and objective information across linear and pairwise coefficients. That can be more compact, especially when the original problem is already pairwise binary. Conversely, the direct graph compilers of \Cref{sec:graphcompiler,sec:csp} have no penalty scale to balance. Factor-compiled hard constraints do reintroduce a penalty scale (\Cref{rem:nopruning}); \Cref{rem:offsets} gives a sufficient choice. Fixed normalized conflict coefficients are also available for the exact-cover QUBO in \eqref{eq:exactcover-qubo} and the clique QUBO below.

\subsection{Connectivity and baseline removal}

Literal matrix density is not always the correct physical cost. As in \Cref{subsec:relative}, the dense baseline in $Q_\tau(G)$ may be removed, leaving the attractive gauge $\widetilde Q_\tau(G)$ of \eqref{eq:relativeQ}, which is sparse when $G$ is sparse; the repulsive gauge $Q_\tau(G)$ instead uses nonnegative interactions and highlights conflicts. A hardware report should state which gauge is implemented and whether the normalization circuitry implements a global all-to-all baseline at negligible marginal cost.

In factor-state graphs, nonzero off-diagonal blocks occur within each factor and between overlapping factors. This preserves bounded interaction degree when domain sizes, arity, and incidence are uniformly bounded. By contrast, a one-hot QUBO often introduces dense penalty blocks inside each assignment group and additional objective couplings. The one-hot baseline is itself a choice: domain-wall encoding uses $n(k-1)$ variables for $n$ variables with $k$ values, with a different penalty connectivity \citep{chancellor-berwald-domain-wall-2022,dominguez-encoding-independent-2023,codognet-comparing-qubo-models-2025}, so the counts in \Cref{tab:comparison} should be read against the QUBO encoding actually used. Neither representation dominates uniformly.

\subsection{Dimension}

\Cref{tab:comparison} summarizes typical state counts.

\begin{table}[t]
\centering
\small
\caption{Representative size behavior. Counts are formulation-level and exclude hardware embedding auxiliaries.}
\label{tab:comparison}
\begin{tabularx}{\linewidth}{>{\raggedright\arraybackslash}p{2.5cm}>{\raggedright\arraybackslash}p{3.1cm}>{\raggedright\arraybackslash}p{3.1cm}X}
\toprule
Problem class & Direct StQP coordinates & Typical binary variables before hardware embedding & Comment \\
\midrule
Clique / independent set & $n$ & $n$ & Same state count; different geometry; the fixed QUBO below has the same clique supports \\
3-SAT with $m$ clauses & at most $3m$ & original Boolean variables plus quadratization auxiliaries, when needed & Which is smaller depends on clause-to-variable ratio \\
$k$-coloring & $nk$ & $nk$ in one-hot QUBO & Similar lifted state count \\
Hamiltonian cycle & $n^2$ & $n^2$ in standard one-hot QUBO & Similar assignment lift; constraint encoding differs \\
Sparse QUBO with $n$ variables, $m$ pairs & at most $2n+4m$ & $n$ & Generic StQP compiler expands local configurations \\
Set packing with $m$ candidate sets & $m$ & $m$ & Direct compatibility is equally compact \\
Exact cover with $m$ candidate sets & $m$ & $m$ & Same candidate count and conflict graph; compare normalized unary weights with reciprocal diagonal weights; see \eqref{eq:exactcover-qubo} \\
Max-Cut & $4|E|$ by factor-state route & $|V|$ & QUBO/Ising is normally much more compact \\
\bottomrule
\end{tabularx}
\end{table}

The remaining metrics of \Cref{subsec:metrics} depend on both the route and the chosen gauge. For the literal integer graph matrix, $\mathcal A_Q\subseteq\{0,1,2\}$ and $\kappa_{\mathrm{coef}}\le2$, with equality in the latter when a nonedge exists. The nonzero count is $M_Q=N_\Delta+\#\{\text{nonedges}\}$; after baseline removal it is $N_\Delta+|E|$. Factor-state matrices satisfy $\mathcal A_Q\subseteq\{0,1\}\cup\{\tau w_{\min}/w_i\}$. Their diagonal range is $w_{\max}/w_{\min}$, and their full nonzero range is $w_{\max}/(\tau w_{\min})$ when a unit nonedge entry exists; otherwise it is $w_{\max}/w_{\min}$. Equation~\eqref{eq:factorcouplings} bounds $M_Q$, while $d_Q$ depends on local-state counts and factor overlaps.

StQP is therefore not a general compression of QUBO. The basic factor-state compiler expands local configurations, whereas clique, set packing, and exact cover admit direct simplex formulations with one coordinate per original candidate.

\subsection{Precision and dynamic range}

The integer graph compiler has a fixed coefficient alphabet, but distinguishing clique sizes $k$ and $k-1$ by their optimal values requires resolving a difference of order $k^{-2}$. For weighted matrices, the diagonal range is $w_{\max}/w_{\min}$; the full coefficient range also depends on the presence of unit nonedge entries. Reciprocal recovery compresses assignment values. The local-value and decision-class bounds of \Cref{subsec:objective-separation} quantify this compression, with refinements when energy spacings are known. A hardware comparison should report:
\begin{enumerate}[label=(\roman*)]
\item programmed coefficient bit depth and realized coefficient error;
\item objective-estimation error or sample variance;
\item feasibility error $|\one^{\mathsf T}x-1|+\norm{x_-}_1$;
\item support-recovery threshold and stability;
\item the relevant decision or local-value separation, identifying whether it is measured exactly or bounded conservatively.
\end{enumerate}
An accurate support can remain useful even when objective readout is insufficient for a threshold decision. Conversely, an accurate scalar objective does not guarantee clean support recovery. Both outputs should be retained.

\subsection{Landscape and local solvers}
\label{subsec:landscape-comparison}

The regularized graph compiler makes landscape complexity explicit: maximal cliques are strict local minima, and by \Cref{thm:weightedlandscape} the same holds with weights. In QUBO, penalty choices can create additional infeasible local minima whose interpretation is less direct \citep{verma-lewis-penalty-2022}. The comparison is not symmetric, however. On any factor-compiled instance every assignment is a strict local minimizer (\Cref{cor:allassignments}), whereas the local minima of the corresponding QUBO under single-bit flips are only the $1$-flip-stable assignments, a subset. On these routes the StQP landscape therefore filters less, not more. Together the two facts suggest a classification of traps:

\begin{description}[leftmargin=2.8cm,style=nextline]
\item[Feasibility traps] local solutions that violate the intended discrete constraints;
\item[Combinatorial traps] valid but suboptimal discrete solutions, such as maximal nonmaximum cliques;
\item[Continuous traps] local minima whose support does not represent a compatible discrete selection;
\item[Embedding traps] artifacts introduced by hardware connectivity or auxiliary chains.
\end{description}

For $0<\tau<1$, the graph construction excludes nonclique \emph{local minima}, but not nonclique stationary points. For the two-vertex edgeless graph, $Q=\left(\begin{smallmatrix}\tau&1\\1&\tau\end{smallmatrix}\right)$ and $x=(1/2,1/2)$ satisfy the simplex KKT conditions, yet $x$ is a strict maximum. Thus stationarity alone does not imply a clique-supported solution. On complete factor-state instances every local minimum represents an assignment, but infeasible and suboptimal assignments remain local minima by \Cref{cor:allassignments}.

\paragraph{A fixed-coefficient QUBO baseline.}
For unweighted clique, the standard conflict-penalty construction, equivalently independent set in the complement graph \citep[Section~4.2]{Lucas2014}, permits
\begin{equation}
  H(z)=-\sum_i z_i+2\sum_{\substack{i<j\\\{i,j\}\notin E}}z_i z_j,
  \qquad z\in\{0,1\}^n.
  \label{eq:cliquequbo-baseline}
\end{equation}
Deleting a selected vertex with $d\ge1$ selected nonneighbors changes $H$ by $1-2d<0$. A nonmaximal clique admits an addition of cost $-1$, whereas every single-bit move from a maximal clique increases $H$. Hence the global minima encode maximum cliques, and the strict single-bit local minima encode exactly maximal cliques. The polynomial coefficients are fixed: $-1$ on unary terms and $2$ on nonedge terms. Fixed coefficients and meaningful clique local minima therefore do not distinguish StQP from every QUBO formulation. The relevant differences include continuous geometry, coefficient signs, readout precision, and the implementation of normalization and interactions.

\subsection{Decoding and verification}

A direct graph-coded StQP candidate can be processed as follows:
\begin{enumerate}[label=(\arabic*)]
\item project or renormalize to the simplex if needed;
\item rank coordinates by mass;
\item extract a compatible support using the graph;
\item greedily extend it to a maximal clique;
\item evaluate the resulting discrete objective;
\item compare against a certified upper or lower bound when available.
\end{enumerate}
This ranking-based procedure is heuristic. The objective-nonincreasing purification in \Cref{app:decoding}, followed by face minimization and extension, is the decoder required for the guarantees in \Cref{subsec:approximation-distortion,subsec:approximation-theory}; first-order and active-set methods for StQP provide the continuous half \citep{bomze-rinaldi-zeffiro-activeset-fw-2020,beretta-torcinovich-pelillo-kkt-2023}. The continuous vector therefore also serves as a heuristic score over local choices. In weighted factor-state models, consistency repair combines local states that agree on shared variables and then locally reoptimizes conflicting factors.

\subsection{Summary of representation costs}

\begin{thesisbox}
The basic factor-state compiler expands an existing pairwise binary objective into local configurations; it can therefore be substantially larger than a direct QUBO, as in the Max-Cut example. For clique and set packing, direct StQP and binary formulations use the same number of coordinates or variables. Both can have fixed coefficients and maximal-clique local minima, as \eqref{eq:cliquequbo-baseline} shows. For exact cover, the weighted independent-set QUBO in \eqref{eq:exactcover-qubo} also avoids elementwise equality penalties and uses the same candidate count and conflict graph.

The case for StQP is its alignment with a normalized, nonnegative native state and its explicit support-based decoding. Whether that alignment offsets continuous precision, readout, and search costs depends on the device and instance family. Comparisons should use strong problem-specific binary formulations and report dimension, connectivity, coefficient range, decision separation, decoding overhead, and end-to-end time-to-target. The formulations establish these comparison points; they do not select a universally superior model.
\end{thesisbox}

\section{A hardware-neutral simplex quadratic optimizer/annealer}
\label{sec:hardware}

Analog and oscillator-based platforms for continuous optimization are advancing quickly \citep{kalinin-aoc-nature-2025,cai-oscillator-np-2025}. The formulations motivate an abstract hardware interface. We call a device satisfying this interface an StQA (\Cref{sec:intro}).

\subsection{Input contract}

The logical input is a symmetric matrix $Q\in\mathbb S^N$, represented by:
\begin{itemize}
\item diagonal onsite terms $Q_{ii}$;
\item programmable relative pair interactions $Q_{ij}$;
\item a declared global scaling and additive simplex gauge;
\item coefficient precision, range, and sign restrictions.
\end{itemize}
For graph-coded NP instances, a minimal device mode needs only diagonal values and binary compatibility or incompatibility couplings. Weighted factor instances require programmable positive diagonal ratios.

\subsection{State contract}

The physical state must map to
\begin{equation}
  x_i\ge0,
  \qquad
  \sum_{i=1}^N x_i=1.
  \label{eq:statecontract}
\end{equation}
Normalization may be enforced exactly by a conservation law or approximately by feedback. If only approximate, the device should report the unnormalized state so that feasibility correction is accounted for as postprocessing rather than hidden inside the solver.

A scaled simplex $\sum_i x_i=R$ is equivalent after $x=Ry$, but it rescales the quadratic energy by $R^2$. Hardware comparisons must distinguish logical normalization from physical amplitude.

\subsection{Energy contract}

The effective Lyapunov or energy function should be proportional, up to a simplex-constant shift, to
\begin{equation}
  E_Q(x)=x^{\mathsf T}Qx.
  \label{eq:energycontract}
\end{equation}
Pairwise physical interactions may realize $Q$ directly, realize a relative gauge $Q-\beta\one\one^{\mathsf T}$, or implement a factorized approximation. The realized matrix $Q^{\mathrm{real}}$ should be measured or calibrated, not merely inferred from requested settings.

\subsection{Search contract}

The local minima characterized in \Cref{subsec:landscape-comparison} motivate search mechanisms beyond monotone local descent. Candidate mechanisms include:
\begin{itemize}
\item controlled noise or temperature schedules;
\item nonlocal support moves;
\item multistart and restart policies;
\item parameter continuation in $\tau$ or in coupling strength, the continuous analogue of annealed replication \citep{bomze-annealed-replication-2002,baum-eagon-1967};
\item parallel replicas and exchange;
\item hybrid branch-and-bound or cutting-plane control;
\item a mechanism that produces certified lower bounds.
\end{itemize}
None of these mechanisms is analyzed here as a global solution algorithm; its success rate and certification capability require separate evaluation.

\subsection{Output contract}

At minimum, a run should return:
\begin{equation}
  x^{\mathrm{raw}},\quad
  t_{\mathrm{device}},\quad
  \text{programmed }Q,\quad
  \text{estimated realized }Q^{\mathrm{real}},
  \label{eq:outputcontract}
\end{equation}
plus all random seeds or control traces needed for reproducibility. Derived quantities should include
\begin{align}
  r_{\mathrm{feas}}&=|\one^{\mathsf T}x-1|+\norm{x_-}_1,\label{eq:feasres}\\
  r_{\mathrm{PG}}(x;\eta)&=
  \frac{\norm{x-\Pi_{\Delta_N}(x-\eta\nabla E_Q(x))}}{\eta},\label{eq:pgres}\\
  E_{\mathrm{logical}}&=x^{\mathsf T}Qx,\qquad
  E_{\mathrm{realized}}=x^{\mathsf T}Q^{\mathrm{real}}x.\label{eq:twoenergies}
\end{align}
Here $(x_-)_i=\max\{-x_i,0\}$, $\Pi_{\Delta_N}$ is Euclidean projection onto the simplex, and the residual uses a positive step size $\eta$ and the Euclidean norm. For feasible $x$, $r_{\mathrm{PG}}=0$ certifies first-order stationarity only, not local or global minimality (\Cref{subsec:landscape-comparison}). For direct graph instances, report decoded clique size and whether the support is maximal. For factor instances, report assignment consistency and original objective value.

\subsection{Timing and energy accounting}

Results should separate:
\begin{enumerate}[label=(\roman*)]
\item matrix compilation;
\item device programming and calibration;
\item physical evolution time;
\item repeated sampling or restarts;
\item readout;
\item feasibility and support repair;
\item classical verification or certification.
\end{enumerate}
Both native-device and end-to-end times are useful; neither should be reported in place of the other without clarification. Similar care is already necessary for Ising and QUBO platforms, where embedding and chain repair can dominate the logical solve \citep{koenz-embedding-overhead-2021,pelofske-three-generations-2023,tasseff-emerging-potential-qa-2024}.

\subsection{A minimal graph-coded device mode}

The integer-scaled graph family suggests a sharply defined first hardware target:
\begin{equation}
  \min_{x\in\Delta_N}
  \left(
  \sum_i x_i^2
  +4\sum_{\{i,j\}\notin E,\ i<j}x_ix_j
  \right).
  \label{eq:minimaldevice}
\end{equation}
The expression follows from $x^{\mathsf T}\Qhat x$ because each off-diagonal matrix entry contributes twice. A device needs one common diagonal strength and one common repulsive strength on programmable conflict pairs. Its global optimum is $1/\omega(G)$. This family is restrictive enough for clean physical experiments and expressive enough to remain NP-complete.

\begin{caution}[No automatic advantage]
A fixed-alphabet NP-complete input family is an attractive hardware target because it simplifies programming. It does not remove NP-hardness or guarantee efficient global search. Claims of advantage require end-to-end comparison on declared instance distributions, including precision and success probability.
\end{caution}

\subsection{Evaluation metrics}

For stochastic or analog solvers, report time-to-target at several target levels rather than only best-of-run objective. For direct graph instances, natural targets are:
\begin{equation}
  \text{find a clique of size at least }k,
  \qquad
  \text{find a maximum clique},
  \qquad
  \text{certify maximum size}.
  \label{eq:targets}
\end{equation}
These are increasingly demanding and should not be conflated. A solver may rapidly find large maximal cliques but rarely find or certify the maximum.

For independent runs of equal duration with success probability $0<p_s<1$, a time-to-solution estimate at confidence $1-\delta$, where $0<\delta<1$, is \citep{mohseni-mcmahon-byrnes-ising-machines-2022,tasseff-emerging-potential-qa-2024}:
\begin{equation}
  T_{1-\delta}
  =t_{\mathrm{run}}
  \max\set{1,\left\lceil\frac{\log\delta}{\log(1-p_s)}\right\rceil},
  \label{eq:tts}
\end{equation}
For $p_s=1$ one run suffices; for $p_s=0$ no finite number of runs attains positive success confidence. Programming, readout, and postprocessing costs should be reported separately.

\section{Limitations and open research questions}
\label{sec:limitations}

\subsection{Universality can hide poor formulations}

Composing a reduction to CLIQUE with the graph compiler proves universality but may destroy application structure. Even a polynomial expansion can be impractical. A useful formulation theory must therefore seek direct encodings, low-rank or sparse matrix representations, and problem-specific decoders rather than treating universal reducibility as the endpoint.

\subsection{Continuous precision is a computational resource}

A real-valued state should not be treated as carrying infinite information at zero cost. The exact graph compiler uses simple input coefficients, but its objective levels approach one another as solution size grows. Weighted compilers may introduce large diagonal dynamic ranges and reciprocal value recovery. Future work should characterize bit complexity, noise sensitivity, and finite-temperature behavior for each formulation family.

\subsection{Approximation theory}
\label{subsec:approximation-theory}

The direct clique formulations preserve multiplicative approximation
under objective-nonincreasing decoding. Indeed, let
$f(x)=x^{\mathsf T}Q_\tau(G)x$ and suppose a feasible point satisfies
$f(x)\le\rho\,\nu(Q_\tau(G))$ for some $\rho\ge1$.
Purification followed by minimization on the resulting clique face
produces a clique $C$ such that
\[
  \frac{\tau}{|C|}
  \le f(x)
  \le\rho\,\frac{\tau}{\omega(G)},
  \qquad\text{hence}\qquad
  |C|\ge\frac{\omega(G)}{\rho}.
\]
Conversely, the barycenter of a clique of size at least
$\omega(G)/\rho$ is a $\rho$-approximate StQP solution.
The weighted formulation gives the same statement with
$|C|$ and $\omega(G)$ replaced by $w(C)$ and $\omega_w(G)$.
Positive scaling, including the integer-scaled matrix $\widehat Q$,
preserves these ratios; adding a simplex-constant objective shift does
not generally preserve them.

Consequently, for every fixed $0<\varepsilon<1$, finding a feasible
$N^{1-\varepsilon}$-approximate solution of the restricted $N$-coordinate
StQP family with diagonal entries one and off-diagonal entries in
$\{0,2\}$ is NP-hard, by the corresponding maximum-clique
inapproximability result~\cite{zuckerman-extractors-2006}.
A polynomial-time algorithm with this guarantee would imply
$\mathrm P=\mathrm{NP}$. This statement concerns multiplicative error
relative to the positive minimum value in the specified normalization.

A different guarantee is available for general StQP. Bomze and de
Klerk~\cite{BomzeDeKlerk2002} establish a polynomial-time approximation
scheme in objective-range-normalized additive error. For rational input data and every fixed
$\varepsilon>0$, a feasible point can be found in polynomial time with
\[
  f(x)-f_{\min}
  \le\varepsilon\bigl(f_{\max}-f_{\min}\bigr),
  \qquad
  f_{\min}:=\min_{y\in\Delta_N}f(y),\quad
  f_{\max}:=\max_{y\in\Delta_N}f(y).
\]
Here the dimension is variable, not fixed. The guarantee is invariant
under positive affine changes of the objective, but it does not bound
$f(x)/f_{\min}$ when the minimum is small compared with the objective
range. There is therefore no conflict with the clique hardness result.

For the direct graph formulation, an additive guarantee
$f(x)\le\tau/\omega(G)+\delta$, with $\delta\ge0$, gives
\[
  |C|\ge
  \frac{\omega(G)}{1+\delta\,\omega(G)/\tau}
\]
after the same decoding. Fixed additive error thus gives a progressively
weaker relative guarantee as the clique number grows. For factor-state
models, \Cref{subsec:approximation-distortion} instead bounds the decoded energy error through
$B-E^\star$; no uniform multiplicative guarantee for the original energy
follows without additional assumptions. Likewise, the binary-linear
feasibility route preserves exact decisions, but an approximate residual
minimum need not satisfy the original constraints. Approximation claims
must therefore specify the original objective, its normalization, and
the decoder; exact representability alone does not determine them.

\subsection{Local landscape beyond cliques}

The local-minimum correspondences in \Cref{thm:weightedlandscape,cor:allassignments} specify the admissible endpoints of local descent but do not determine their basins of attraction. The unresolved quantities include basin volumes under a declared initialization distribution, convergence times, the quality of decoded assignments, and their dependence on offsets and search dynamics. The existence of a local minimum for every assignment does not imply that all assignments are reached with equal probability.

\subsection{Alternative simplex lifts}
\label{subsec:alternative-lifts}

Relaxation labeling assigns a probability vector to each object and optimizes compatibility between labels \citep{Pelillo1997Relaxation}. With variables indexed by $i$ and labels $a\in D_i$, its state space has the form
\[
  \mathcal P=\left\{(p_{ia}):p_{ia}\ge0,\quad
       \sum_{a\in D_i}p_{ia}=1\ \text{for every }i\right\}
       =\prod_i\Delta_{|D_i|}.
\]
Quadratic optimization over products of simplexes is a related established class \citep{bomze-schachinger-multistqp-2010}. Each block has a prescribed mass, whereas the single-simplex constructions here allow mass to move between blocks. Rescaling $x_{ia}=p_{ia}/n$ embeds $\mathcal P$ into a single simplex only with the additional constraints $\sum_a x_{ia}=1/n$. Dropping those constraints changes the feasible set. Consequently the maximal-clique and every-assignment local-minimum results do not automatically transfer to product-simplex formulations.

Other representations use marginal polytopes, low-rank factorizations, or entropy regularization. A useful comparison must account for both the number of local states and the cost of enforcing multiple normalization constraints.

\subsection{Global certification}

StQP has exact copositive and completely positive reformulations \citep{BomzeDurEtAl2000,BomzeDeKlerk2002,Burer2009}, summarized in \Cref{app:copositive}. This suggests hybrid architectures in which an unconventional device supplies primal candidates while classical conic, semidefinite, linear, or branch-and-bound methods supply lower bounds. For graph-coded instances, combinatorial bounds on clique number offer another certification path \citep{PardalosXue1994,WuHao2015,DeKlerkPasechnik2002}, and exact StQP solvers based on mixed-integer and conic reformulations give a second \citep{yildirim-milp-stqp-2021,bomze-locatelli-tardella-bounds-2008,duer-rendl-conic-survey-2021,scozzari-tardella-clique-algorithm-2008}. The division of labor between candidate generation and certification may be more realistic than expecting a physical relaxer to certify globality internally.

\subsection{Formulation search as an optimization problem}

An instance may admit many equivalent matrices related by scaling, gauge shifts, variable duplication, factor grouping, and alternative reduction routes. Selecting a formulation for a given device can itself be posed as a multiobjective optimization problem over:
\[
  N_\Delta,
  \quad M_Q,
  \quad d_Q,
  \quad \kappa_{\mathrm{coef}},
  \quad \text{decision gap},
  \quad \text{decoder robustness}.
\]
Automated formulation search is a promising counterpart to compiler optimization in conventional computing.

\subsection{Empirical question}
\label{subsec:agenda}

This paper provides formulations, proofs, and representation-cost analysis. It does not report solver-performance experiments, basin measurements, or hardware benchmarks. A numerical study could compare projected-gradient or replicator dynamics with random restarts on direct graph and factor-state instances, reporting decoded objective values, success probabilities, and sensitivity to offsets. Comparisons with problem-specific binary solvers would then test whether the continuous representation offers a practical benefit for the chosen instance family.

\section{Conclusion}
\label{sec:conclusion}

The catalog expresses discrete problems as quadratic optimization over one normalized, nonnegative state, using established clique formulations and compatibility graphs. It records explicit matrices, dimensions, thresholds, and decoding rules for direct graph routes and for generic factor-state routes. Twelve of Karp's 21 problems use direct graph formulations; the other nine use factor-state compilation, including six through binary-linear feasibility.

The representation analysis identifies the cost of these choices. Direct graph routes have simple coefficients and support-based decoding. Generic factor-state compilation can expand dimension and creates a strict local minimum for every complete assignment. Reciprocal objective recovery introduces a dependence on weights, offsets, and attainable energy spacings. The resulting separation bounds and decoding guarantees specify accuracy conditions under the stated assumptions.

Strong binary formulations can match the candidate count and use fixed normalized conflict coefficients for clique, packing, and exact cover. The practical comparison therefore depends on the implementation of normalization and interactions, continuous precision, search dynamics, and readout. Establishing performance for any device requires measurements beyond the exact formulations developed here.

\appendix

\section{Compilation procedures}
\label{app:pseudocode}

\begin{algorithm}[H]
\caption{Three-level graph-to-StQP compiler}
\label{alg:graphcompiler}
\noindent\textbf{Input:} A nonempty simple undirected graph $G=(V,E)$ and $0<\tau<1$.\\
\textbf{Output:} A symmetric matrix $Q\in\mathbb S^{|V|}$.
\begin{enumerate}[label=\arabic*.,leftmargin=2em]
\item Set $Q_{ii}=\tau$ for every $i\in V$.
\item For each unordered pair $\{i,j\}$, set $Q_{ij}=Q_{ji}=0$ if it is an edge, and set $Q_{ij}=Q_{ji}=1$ otherwise.
\item Return $Q$.
\end{enumerate}
\end{algorithm}

\begin{algorithm}[H]
\caption{Clause--literal compatibility graph}
\label{alg:satcompiler}
\noindent\textbf{Input:} Nonempty CNF clauses $C_1,\ldots,C_m$, with $m\ge1$.\\
\textbf{Output:} A clause--literal compatibility graph $H$.
\begin{enumerate}[label=\arabic*.,leftmargin=2em]
\item Create one vertex $v_{r,\ell}$ for every literal occurrence $\ell\in C_r$.
\item Join $v_{r,\ell}$ to $v_{s,\ell'}$ exactly when $r<s$ and $\ell'$ is not the complement of $\ell$.
\item Return $H$.
\end{enumerate}
\end{algorithm}

\begin{algorithm}[H]
\caption{Factor-state weighted StQP compiler}
\label{alg:factorcompiler}
\noindent\textbf{Input:} Nonempty finite domains $D_i$, a nonempty collection of complete factor tables $\phi_\alpha$ on scopes $S_\alpha$, offsets $B_\alpha>\max_a\phi_\alpha(a)$, and $0<\tau<1$.\\
\textbf{Output:} A matrix $Q$ and recovery constants $(B,w_{\min},\tau)$.
\begin{enumerate}[label=\arabic*.,leftmargin=2em]
\item For every factor $\alpha$ and assignment $a\in\prod_{i\in S_\alpha}D_i$, create a vertex $v=(\alpha,a)$ with weight $w_v=B_\alpha-\phi_\alpha(a)$.
\item Mark distinct vertices $(\alpha,a)$ and $(\beta,b)$ compatible exactly when $\alpha\ne\beta$ and $a,b$ agree on $S_\alpha\cap S_\beta$.
\item Compute $w_{\min}=\min_v w_v$ and set $Q_{vv}=\tau w_{\min}/w_v$.
\item For $u\ne v$, set $Q_{uv}=0$ on compatible pairs and $Q_{uv}=1$ otherwise.
\item Return $Q$, $B=\sum_\alpha B_\alpha$, $w_{\min}$, and $\tau$.
\end{enumerate}
\end{algorithm}

\section{Detailed decoding and purification}
\label{app:decoding}

\subsection{Purifying an approximate graph solution}

Given $x\in\Delta_n$ and graph $G$, a simple decoder orders vertices by decreasing $x_i$ and greedily builds a clique. A more objective-aware purification repeatedly identifies a nonedge $\{i,j\}$ in the positive support and moves all mass from one endpoint to the other, choosing the endpoint move with smaller StQP objective. By the proof of \Cref{thm:graphcompiler}, this never increases the objective and terminates at clique support in at most $n-1$ eliminations.

Once a clique $C$ is obtained, replace its masses by the exact face minimizer:
\begin{equation}
  x_i\leftarrow\frac1{|C|}\quad(i\in C)
  \label{eq:unweightedrepair}
\end{equation}
for the unweighted compiler, or
\begin{equation}
  x_i\leftarrow\frac{w_i}{\sum_{j\in C}w_j}\quad(i\in C)
  \label{eq:weightedrepair}
\end{equation}
for the weighted compiler. A greedy extension can then make $C$ maximal. None of these operations certifies maximum weight.

\subsection{Factor-state decoding}

At any clique-supported local minimizer, each factor contributes exactly one local state (\Cref{cor:allassignments}), so the selected states agree on overlaps and merge into a complete assignment; only its value distinguishes a global minimizer from the rest. For an approximate nonclique support, one may:
\begin{enumerate}[label=(\roman*)]
\item select the largest-mass state of each factor;
\item resolve variable disagreements by weighted voting;
\item perform local coordinate or factor updates in the original objective.
\end{enumerate}
This repair is heuristic and should be timed and reported separately.

\section{Explicit formulations for selected catalog entries}
\label{app:formulations}

\subsection{Independent set and vertex cover}

This is \Cref{subsec:cliqueis} in matrix form. For a graph $G$, let $A_{\overline G}$ be the adjacency matrix of its complement. Define
\[
  \Qhat_{\mathrm{IS}}=2\one\one^{\mathsf T}-2A_{\overline G}-I.
\]
Then
\[
  \min_{x\in\Delta_n}x^{\mathsf T}\Qhat_{\mathrm{IS}}x=\frac1{\alpha(G)}.
\]
If $x^\star$ has support $S$, then $S$ is a maximum independent set and $V\setminus S$ is a minimum vertex cover.

\subsection{Coloring matrix blocks}

Index StQP coordinates by $(v,c)$. The integer-scaled matrix has diagonal one. For distinct coordinates, the coefficient is zero exactly when they are compatible:
\begin{equation}
  \widehat q_{(v,c),(u,d)}=0
  \quad\Longleftrightarrow\quad
  v\ne u\ \text{and}\
  \bigl(\{u,v\}\notin E\ \text{or}\ c\ne d\bigr).
  \label{eq:colorblocks}
\end{equation}
All other off-diagonal coefficients equal two. The same-vertex $k\times k$ blocks are conflict cliques in the matrix, and an edge $\{u,v\}$ contributes conflicts only between equal-color pairs.

\subsection{Hamiltonian matrix blocks}

Coordinates are $(p,v)$. Same-position and same-vertex pairs receive coefficient two. For consecutive positions, a pair receives zero only if the corresponding graph edge exists. For nonconsecutive positions with different vertices, the coefficient is zero. This block structure permits procedural matrix evaluation without materializing all $n^4$ potential assignment pairs.

\subsection{Max-Cut local states}

For edge $e=\{i,j\}$, the four local states are $00,01,10,11$. Their weights are $L,L+c_e,L+c_e,L$. States on disjoint edges are always compatible. States on incident edges are incompatible exactly when they assign different labels to the shared endpoint. Thus each pair of incident edge factors contributes a sparse $4\times4$ conflict block.

\section{A note on copositive certification}
\label{app:copositive}

For any $Q\in\mathbb S^N$,
\begin{equation}
  \nu(Q)=\max\set{\lambda:Q-\lambda\one\one^{\mathsf T}\in\mathcal{COP}_N},
  \label{eq:copositive}
\end{equation}
where $\mathcal{COP}_N$ is the copositive cone \citep{BomzeDurEtAl2000,BomzeDeKlerk2002}. Indeed, $Q-\lambda\one\one^{\mathsf T}$ is copositive if and only if
\[
  y^{\mathsf T}Qy\ge\lambda(\one^{\mathsf T}y)^2
  \qquad\forall y\ge0.
\]
Normalizing any nonzero $y\ge0$ to $x=y/(\one^{\mathsf T}y)\in\Delta_N$ gives $x^{\mathsf T}Qx\ge\lambda$. This identity provides an exact global lower-bound certificate, although optimizing over the copositive cone is itself computationally hard. Tractable inner approximations, semidefinite hierarchies, and combinatorial bounds can nevertheless certify many instances.

For a graph-coded matrix $\Qhat(G)$, proving $\omega(G)\le k$ is equivalent to proving
\begin{equation}
  \Qhat(G)-\frac1k\one\one^{\mathsf T}\in\mathcal{COP}_n.
  \label{eq:cliquecopositive}
\end{equation}
Thus the formulation connects the physical energy minimum directly to a classical conic certificate.

\section{Notation summary}
\label{app:notation}

\begin{center}
\small
\begin{tabularx}{0.92\linewidth}{>{\raggedright\arraybackslash}p{3.0cm}X}
\toprule
symbol & meaning \\
\midrule
$\Delta_N$ & probability simplex $\{x\ge0:\one^{\mathsf T}x=1\}$ \\
$\nu(Q)$ & optimum of $\min_{x\in\Delta_N}x^{\mathsf T}Qx$ \\
$A_G$ & adjacency matrix of graph $G$ \\
$\omega(G)$ & maximum clique cardinality \\
$\omega_w(G)$ & maximum total vertex weight of a clique \\
$Q_\tau(G)$ & three-level graph StQP matrix with diagonal $\tau$ \\
$\Qhat(G)$ & integer-scaled graph matrix $2Q_{1/2}(G)$ \\
$x(C)$ & uniform barycenter of clique $C$ \\
$N_\Delta$ & number of simplex coordinates in an encoding \\
$M_Q$ & number of nonzero upper-triangular matrix entries \\
$E(z)$ & original finite-domain factor objective \\
$B_\alpha$ & factor offset satisfying \eqref{eq:Balpha} \\
$w_{\alpha,a}$ & factor-state vertex weight $B_\alpha-\phi_\alpha(a)$ \\
$P(z,y)$ & zero-penalty QUBO for binary-linear feasibility \\
\bottomrule
\end{tabularx}
\end{center}


\section*{Competing interests}
The authors are affiliated with Quantum Computing Inc. (QCi).

\renewcommand{\bibfont}{\small}
\setlength{\bibsep}{4pt plus 1pt minus 1pt}
\setlength{\emergencystretch}{1em}
\bibliographystyle{unsrtnat}
\bibliography{references}

\end{document}